%% file: paper.tex
\documentclass[a4paper,11pt]{article}
\usepackage{geometry}
\usepackage{hyperref}
\usepackage[svgnames]{xcolor}
\hypersetup{colorlinks={true},urlcolor={blue},linkcolor={DarkBlue},citecolor=[named]{DarkGreen}}
\usepackage[authoryear,square]{natbib}

\usepackage{xspace}
\usepackage{url}
\usepackage[small]{caption}
\usepackage{graphicx}
\usepackage{subcaption}
\usepackage{amsmath}
\usepackage{amsthm}
\usepackage{booktabs}
\usepackage[most]{tcolorbox}

\usepackage{tikz}  
\usetikzlibrary{arrows}
\usetikzlibrary{patterns,snakes}
\usetikzlibrary{decorations.shapes}
\tikzstyle{overbrace text style}=[font=\tiny, above, pos=.5, yshift=5pt]
\tikzstyle{overbrace style}=[decorate,decoration={brace,raise=5pt,amplitude=3pt}]
\usetikzlibrary{shapes.geometric}

\usepackage[linesnumbered,ruled,vlined]{algorithm2e}
\RequirePackage{comment}
\usepackage{authblk}
\usepackage{mathtools}
\RequirePackage{comment}
\usepackage{xifthen}
\usepackage{arydshln}
\usepackage[tt=false]{libertine}
\usepackage[capitalise]{cleveref} 
\usepackage{amssymb}
\usepackage{dsfont}
\usepackage{colortbl}
\usepackage{graphics}
\usepackage{float}
\usepackage{color}
\usepackage{array}
\usepackage{multirow}
\usepackage[shortlabels]{enumitem}
\usepackage{amsmath,latexsym,color,amsthm,authblk}

\newcommand{\VV}{\ensuremath{\mathcal V}\xspace}

\newtheorem{theorem}{\bf Theorem}
\newtheorem{lemma}{\bf Lemma}

\newtheorem{definition}{\bf Definition}

\crefname{theorem}{theorem}{\bf Theorem}
\crefname{example}{example}{\bf Example}
\crefname{observation}{observation}{\bf Observation}
\crefname{lemma}{lemma}{\bf Lemma}
\crefname{corollary}{corollary}{\bf Corollary}
\crefname{proposition}{proposition}{\bf Proposition}
\crefname{definition}{definition}{\bf Definition}
\crefname{claim}{claim}{\bf Claim}
\crefname{reductionrule}{reduction rule}{\bf Reduction rule}
\crefname{assumption}{assumption}{\bf Assumption}

\DeclareMathOperator*{\argminaa}{arg\,min}

\newcommand{\instance}{\ensuremath{I}\xspace}

\newcommand{\lb}{\left(}
\newcommand{\rb}{\right)}
\DeclarePairedDelimiter\ceil{\lceil}{\rceil}

\newcommand{\bud}{\ensuremath{B}\xspace}
\newcommand{\proj}{\ensuremath{\mathcal{T}}\xspace}
\newcommand{\voters}{\ensuremath{N}\xspace}

\newcommand{\approf}{\ensuremath{\mathcal{A}}\xspace}

\newcommand{\cof}[1]{\ensuremath{\operatorname{cost \hspace{0.5mm}\!}\left(#1\right)}\xspace}
\newcommand{\co}{\ensuremath{\operatorname{cost \hspace{0.5mm}\!}}\xspace}
\newcommand{\scoreof}[1]{
    \ifthenelse{\isempty{#1}}
        {\ensuremath{score(p)}\xspace}
        {\ensuremath{score(#1)}\xspace}
}
\newcommand{\aof}[1]{
    \ifthenelse{\isempty{#1}}
        {\ensuremath{A(i)}\xspace} 
        {\ensuremath{A(#1)}\xspace} 
}
\newcommand{\leftof}[1]{
    \ifthenelse{\isempty{#1}}
        {\ensuremath{b_i}\xspace} 
        {\ensuremath{b_{(#1)}}\xspace} 
}
\newcommand{\affordable}{\ensuremath{\rho}\xspace}
\newcommand{\utof}[2]{
    \ifthenelse{\isempty{#1}}
    {
        \ifthenelse{\isempty{#2}}
            {\ensuremath{u_i\left(S\right)}\xspace}
            {\ensuremath{u_i\left(#2\right)}\xspace}
    }
    {
        \ifthenelse{\isempty{#2}}
            {\ensuremath{u_{#1}\left(S\right)}\xspace}
            {\ensuremath{u_{#1}\left(#2\right)}\xspace}
    }
}
\newcommand{\cardof}[1]{
    \ifthenelse{\isempty{#1}}
        {\ensuremath{|S|}\xspace}
        {\ensuremath{|#1|}\xspace}
}
\newcommand{\appof}[2]{
    \ifthenelse{\isempty{#1}}
    {
        \ifthenelse{\isempty{#2}}
            {\ensuremath{A_i(j)}\xspace}
            {\ensuremath{A_i(#2)}\xspace}
    }
    {
        \ifthenelse{\isempty{#2}}
            {\ensuremath{A_{#1}(j)}\xspace}
            {\ensuremath{A_{#1}(#2)}\xspace}
    }
}
\newcommand{\priceof}[2]{
    \ifthenelse{\isempty{#1}}
    {
        \ifthenelse{\isempty{#2}}
            {\ensuremath{p_i(p)}\xspace}
            {\ensuremath{p_i(#2)}\xspace}
    }
    {
        \ifthenelse{\isempty{#2}}
            {\ensuremath{p_{#1}(p)}\xspace}
            {\ensuremath{p_{#1}(#2)}\xspace}
    }
}

\newcommand{\pbrule}{\ensuremath{f}\xspace}

\newcommand{\winnerfunction}{\ensuremath{\Phi}\xspace}
\newcommand{\ruleof}[2]{%
	\ifthenelse{\isempty{#2}}%
	{\ensuremath{\operatorname{#1\!}\left(\approf\right)}\xspace}
	{\ensuremath{\operatorname{#1\!}\left(#2\right)}\xspace}
}
\newcommand{\degreeof}[2]{
    \ifthenelse{\isempty{#1}}
    {
        \ifthenelse{\isempty{#2}}
            {\ensuremath{d_{\pbrule}\left(l\right)}\xspace}
            {\ensuremath{d_{\pbrule}\left(#2\right)}\xspace}
    }
    {
        \ifthenelse{\isempty{#2}}
            {\ensuremath{d_{#1}\left(l\right)}\xspace}
            {\ensuremath{d_{#1}\left(#2\right)}\xspace}
    }
}
\newcommand{\winners}[2]{%
  \ifthenelse{\isempty{#2}}%
    {\ensuremath{\operatorname{\winnerfunction\!}\left(#1,\instance\right)}\xspace}
    {\ensuremath{\operatorname{\winnerfunction\!}\left(#1,\instance#2\right)}\xspace}
}

\makeatletter
\newcommand\mathcircled[1]{%
  \mathpalette\@mathcircled{#1}%
}
\newcommand\@mathcircled[2]{%
  \tikz[baseline=(math.base)] \node[draw,circle,inner sep=1pt] (math) {$\m@th#1#2$};%
}
\makeatother
\makeatletter
\newcommand{\thickhline}{%
    \noalign {\ifnum 0=`}\fi \hrule height 1.3pt
    \futurelet \reserved@a \@xhline
}
\newcommand{\toothickhline}{%
    \noalign {\ifnum 0=`}\fi \hrule height 2pt
    \futurelet \reserved@a \@xhline
}

\newcommand{\mes}{\textsc{mes}\xspace}

\makeatletter
\newcommand*{\inlineequation}[2][]{%
  \begingroup
    \refstepcounter{equation}%
    \ifx\\#1\\%
    \else
      \label{#1}%
    \fi
    \relpenalty=10000 %
    \binoppenalty=10000 %
    \ensuremath{%
      #2%
    }%
    ~\@eqnnum
  \endgroup
}
\makeatother

\newcolumntype{"}{@{\hskip\tabcolsep\vrule width 1.3pt\hskip\tabcolsep}}
\newcolumntype{x}{@{\hskip\tabcolsep\vrule width 2pt\hskip\tabcolsep}}
\makeatother

\makeatletter
\newcommand*{\rom}[1]{\expandafter\@slowromancap\romannumeral #1@}
\makeatother

\title{\bf Proportionality Degree in Participatory Budgeting}
\author[1]{Aris Filos-Ratsikas}
\author[2]{Sreedurga Gogulapati}
\author[1]{Georgios Kalantzis}

\affil[1]{University of Edinburgh, UK}
\affil[2]{Indian Institute of Technology Hyderabad, India}
\date{}

\begin{document}
\maketitle

\begin{abstract}
We initiate the study of the proportionality degree for participatory budgeting, with a particular focus on two popular methods: the Method of Equal Shares (\mes) and Phragm\'{e}n's Sequential Rule. Among other results, we derive tight bounds (up to small constant factors) on the proportionality degree of these two rules, which showcase that, despite \mes satisfying stronger axiomatic guarantees, the two rules have the same proportionality degree from a quantitative perspective. We complement our theoretical findings with an extensive experimental evaluation on real-world participatory budgeting datasets, the results of which closely mirror those of our developed theory. Our experiments also provide more insights into the comparisons between the rules.
\end{abstract}
\section{Introduction}
Participatory budgeting (PB) is a democratic process through which citizens directly influence the allocation of public funds. First implemented in Porto Alegre in 1989, PB has since been adopted globally at various scales, from neighborhoods and cities to schools and universities, as a mechanism to enhance civic engagement, improve transparency, and promote more equitable outcomes in public spending \citep{de2022international,wampler2021participatory}. At its core, PB enables participants to propose, deliberate upon, and select projects proposals for funding within a fixed budget constraint, thereby embedding elements of direct democracy into budgetary decision-making.

The computational aspects of PB have received increasing attention in recent years, motivated by the need to design fair and efficient methods for aggregating citizen preferences into concrete budget allocations \citep{aziz2021participatory}. A significant line of work in this area adopts an axiomatic approach, where desirable properties of fairness and representation are formalized as axioms. In particular, axioms inspired by approval-based multi-winner voting such as Justified Representation (JR), Proportional Justified Representation (PJR), and Extended Justified Representation (EJR) \citep{aziz2017justified, sanchez2017proportional}, have been adapted to the PB setting to capture the normative idea that cohesive groups of voters should be entitled to a share of the budget roughly proportional to their size \citep{peters2021proportional}. These axioms offer qualitative guarantees and serve as important design goals for PB rules. However, axiomatic properties are inherently binary: a rule either satisfies a given axiom or it does not.

To complement the axiomatic approach, the notion of the \textit{proportionality degree} has been introduced as a quantitative measure of how well a participatory budgeting rule achieves proportional representation \citep{aziz2018complexity, skowron2021proportionality}. The proportionality degree captures the extent to which a cohesive group of voters can secure a share of the budget in proportion to its size.
This notion has been extensively studied in the context of multi-winner voting, where the proportionality degree of prominent rules such as Phragmén’s rule \citep{Phragmen1894methode} and the Proportional Approval Voting (PAV) rule \citep{thiele1895om} (see also \citep{kilgour2010approval}) is well understood. However, despite the growing body of work on PB, the proportionality degree of PB rules has remained unexplored.

This gap motivates our study: we aim to initiate a systematic examination of the proportionality degree of all the PB rules known to satisfy proportionality-related axioms. We particularly focus on the well known \textit{Method of Equal Shares} \citep{peters2021proportional} and the \textit{ Phragm\'{e}n} sequential \citep{Phragmen1894methode} (see also \citep{los2022proportional}) rules. Both rules share a market-based intuition and have been recognized for their fairness-oriented properties in the literature. 

\subsection{Our Contribution}
We derive tight lower and upper bounds (up to small constant factors) on the proportionality degree of these two rules. \mes is known to satisfy strong axiomatic guarantees such as Extended Justified Representation (EJR), whereas Phragmén does not. However, somewhat surprisingly, both rules have the same proportionality degree guarantees from a quantitative perspective. This finding highlights a deeper structural similarity between these methods.

To complement our theoretical analysis, we also conduct an extensive experimental evaluation on real-world PB datasets covering 100 elections. The experiments include studying the proportionality degree of the \mes and Phragm\'{e}n rules, both with and without exhaustion. Additionally, we also compare their performance with a rule that greedily selects projects with respect to their approval score, which is widely used in practice. Our empirical results show that, with respect to proportionality degree, both the \mes and Phragmén rules with exhaustion exhibit comparable performance in practice (with the latter performing slightly better than the former). Moreover, all the above rules clearly outperform the greedy rule.

\subsection{Related Work}
The Participatory Budgeting literature focuses broadly on two objectives: welfare \citep{talmon2019framework,laruelle2021voting,sreedurga2022maxmin} and fairness.
The goal of fairness in PB is to ensure that the satisfaction of voters meets a certain threshold. There are various notions of fairness such as core-based notions \citep{fain2016core,fain2018fair,peters2021proportional,maly2023core}, share-based notions \citep{bogomolnaia2005collective,duddy2015fair,aziz2018rank,aziz2019fair,sreedurga2022indivisible,sreedurga2023individual,airiau2023portioning,sreedurga2024hybrid}, and proportionality-based notions \citep{aziz2018proportionally,skowron2020participatory,peters2021proportional,aziz2021proportionally,fairstein2021proportional,los2022proportional,fairstein2022welfare,peters2020proportionality}. Note that, proportionality is the most widely studied fairness concept among the three, especially due to its computational and cognitive simplicity. The key idea behind these proportionality notions is that, if a group of voters unanimously supports a set of projects, the group is guaranteed a satisfaction proportional to its size. Such groups of voters are said to be cohesive.

\cite{aziz2018proportionally} first studied the concept of proportionality in PB, inspired by the similar concept in multi-winner voting (MWV), a special case of PB in which every project has unit cost and the budget allows $k$ projects to be funded \citep{lackner2023multi}. Consequently, much of the PB research on proportionality has drawn extensively on the MWV literature. Classical proportional MWV rules such as the Sequential Phragmén rule \citep{Phragmen1894methode,brill2024Phragmen} and the Method of Equal Shares (MES) \citep{peters2020proportionality} have been extended to PB, respectively by \cite{los2022proportional} and \cite{peters2021proportional}. The well known Proportional Approval Voting (PAV) rule \citep{thiele1895om} in the MWV literature, however, is proved to not satisfy proportionality guarantees in PB \citep{peters2021proportional,los2022proportional}.
\mes is polynomial-time computable and satisfies Extended Justified Representation (EJR) \citep{peters2021proportional}, whereas Phragmén satisfies Proportional Justified Representation (PJR) but fails EJR \citep{los2022proportional}. MES has also been extended to settings with additive utilities; comparable extensions for Phragmén remain open.

In MWV, the proportionality degree offers a fine-grained yardstick for how closely a rule approximates an ideal proportionality. Tight bounds have been established for PAV, Phragmén, and related rules \citep{skowron2021proportionality,aziz2018complexity}, with PAV and its local-search variants \citep{aziz2018complexity,kraiczy2024lower} achieving the optimal degree. In addition, in the case of multiwinner voting a general bound is known for all rules which satisfy EJR \cite{sanchez2017proportional}, which coincides with the known proportionality degree of Phragmén \cite{skowron2021proportionality}. Thus, as \mes is a rule which satisfies EJR in approval-based multiwinner voting, we know that it achieves proportionality degree at least as large as that of Phragmén, but an upper bound is unknown.

By contrast, no prior work has analyzed the proportionality degree for PB-specific rules. Our paper addresses this question by deriving near-tight bounds for \mes and Phragmén in the PB setting and showing that their proporionality degree guarantees coincide, despite the rules differing in their axiomatic properties (namely \mes's compatibility with EJR versus Phragmén’s incompatibility).

This work assumes that the preferences over the projects are expressed as approval votes. This is indeed the most popular preference elicitation method in practice \citep{benade2018efficiency} and hence is also widely studied in the literature \citep{aziz2018proportionally,talmon2019framework,fluschnik2019fair,goel2019knapsack,baumeister2020irresolute,rey2020designing,jain2020participatoryg,freeman2021truthful,fairstein2021proportional,benade2021preference,laruelle2021voting,los2022proportional,brill2023proportionality}.

\section{Preliminaries}
An approval-based \textit{participatory budgeting (PB)} instance is a tuple $\instance = (N, \proj, B, C, \approf)$, where:
\begin{itemize}[leftmargin=*]
  \item[-] $N$ is a set of $n$ \emph{voters}, and $\proj$ is a set of $m$ \textit{projects}.
  \item[-] $\bud$ is the available \textit{budget}. 
  \item[-] $C$ is the \emph{cost profile}, namely $\co : \proj \rightarrow \mathbb{R}_{> 0}$ is a function that associates a cost with each $t \in \proj$. For each subset of projects $T \subseteq \proj$, we define $\cof{T} = \sum_{t \in T} \cof{t}$ for the total cost of $T$. We will use $\mathcal{C}$ to refer to the collection of all possible cost profiles. 
  \item[-] \approf is the \emph{approval profile}, namely a collection of subsets of projects $A_i \subseteq \proj$ which denote the projects that voter $i \in N$ approves.
\end{itemize}
A subset of projects $W \in \proj$ is feasible if $\cof{W} \leq B$. A participatory budgeting rule (PB rule) $\pbrule$ takes as input an instance of the above form and returns a feasible subset of projects $W = \pbrule(\instance)$, which we refer to as the \emph{proposal} of rule $\pbrule$.

To define the appropriate notions of proportionality, in multiwinner voting the literature has employed the concept of an $\ell$-large group. In the case of PB however, a notion that is based only on the number of projects is inadequate, since the ``proportional rights'' of each project's group of supporters should also depend on the project's cost. For this reason, the following notion of \emph{$T$-cohesiveness} has been defined for PB, see \citep{peters2021proportional,los2022proportional}.

\begin{definition}[$T$-Cohesive Group \citep{peters2021proportional,los2022proportional}]\label{def: tcohesive}
   Given a PB instance \instance and a subset $T \subseteq \proj$ of projects, we say that a group of voters $V \subseteq \voters$ is $T$-cohesive if and only if 
   \[T \subseteq \bigcap\limits_{i \in V}{A_i} \ \ \text{ and } \ \ \frac{\cof{T}}{\bud} \leq \frac{|V|}{n}.\]
\end{definition}

We next define the notion of the \emph{average satisfaction} of a set of voters $V$. 
\begin{definition}[Average Satisfaction]
Given a PB instance \instance, a subset of projects $W$ and a subset of voters $V$, the average satisfaction of $V$ is defined as 
\[
\emph{\text{avg}}_W(V) = \frac{1}{|V|} \sum_{i \in V} |W \cap A_i|.
\]
\end{definition}

We are now ready to define the proportionality degree.
\begin{definition}[Proportionality Degree]
    Given a PB instance \instance, the proportionality degree of a PB rule \pbrule for a subset $T \in \proj$, denoted by $d_{\pbrule}(T)$, is defined as follows:
\[
        d_{\pbrule}(T) = \sup\left\{g(T): \min_{V \in \VV(T)} \emph{\text{avg}}_{f(I)}(V) \geq \min(|T|, g(T)) \right\},     
\]
where $\VV(T)$ denotes the set of all $T$-cohesive groups of voters.
\end{definition}

The proportionality degree of a PB rule \pbrule for a given subset $T \in \proj$ of projects intuitively denotes the worst-case bound on the average satisfaction of $T$-cohesive groups of voters.
In particular, the proportionality degree of a rule for a set $T$ represents the largest worst-case lower bound on the average satisfaction of any $T$-cohesive group of voters.

Finally, we refer to a well-known axiom of proportionality, namely the \emph{Extended Justified Representation (EJR)}. EJR was first proposed for approval-based committee elections by \cite{aziz2017justified} and extended to the participatory budgeting model by \cite{peters2021proportional}. 

\begin{definition}[Extended Justified Representation (EJR) for (approval-based) PB \citep{peters2021proportional}]\label{def: EJR-PB}
    A rule $\pbrule$ satisfies EJR if for each participatory budgeting  instance $I$, each $T \subseteq \proj$, and each $T$-cohesive group $V \subseteq N$ of voters there exists a voter $i \in V$ such that $|A_i \cap \pbrule(I)| \geq |T|$.
\end{definition}

\section{Proportionality Degree of the Method of Equal Shares}
First, we provide the definition of the {\sc Method of Equal Shares} (\mes) for approval based profiles in participatory budgeting. It is worth mentioning that \mes is one of the most widely considered proportional PB rules in theory and practice\footnote{See \url{equalshares.net}.}. It is the only known polynomial-time PB rule that satisfies EJR for approval preferences. It also satisfies additional proportionality notions such as PJR-X and Local-BPJR-L, see  \citep{aziz2018proportionally,brill2023proportionality} for a discussion of those notions.
\begin{tcolorbox}[
    standard jigsaw,
    opacityback=0,  
]
    \begin{definition}[{\sc Method of Equal Shares }\citep{peters2021proportional}]\label{def: mes}
     Each voter is initially given an equal fraction of the budget, i.e. each voter is given $\bud/n$ monetary units. We start with an empty proposal $W = \emptyset$ and sequentially add projects to $W$, one in each round. To add a project $p$ to $W$, we need the voters to pay for $t$. \medskip 
     
     For each selected project $t \in W$ let $\priceof{i}{t}$ be the amount that voter $i$ pays for $t$; obviously, we require $\sum_{i \in N} \priceof{i}{t} = \cof{t}$, to cover the cost of the project. Let $\priceof{i}{W} = \sum_{t \in W} \priceof{i}{t} \leq \bud/n$ be the amount that voter $i$ has paid so far for the projects selected in $W$, and let $\leftof{} = \bud/n - \priceof{}{W}$ be the amount that voter $i$ has left in their allocated budget (clearly, before the end of the first round $\leftof{} = \bud/n$ for all $i \in N$). \medskip 
     
     For $\affordable \geq 0$, we say that a project $t \in \mathcal{T}\setminus W$ is \emph{$\affordable$-affordable} if   
     \[
     \sum_{i \in N_t}\min(\leftof{}, \affordable) = \cof{t}
     \]
     where $N_t \subseteq N$ denotes the set of voters that approve project $t$, i.e., $N_t =\{i \in N: t \in A_i\}$. If no project is $\affordable$-affordable for any \affordable, \mes terminates and returns $W$. Otherwise, it selects a project $t \in \mathcal{T}\setminus W$ that is \affordable-affordable for a minimum \affordable. Individual payments are given by
     \[
     \priceof{i}{t} = \min(\leftof{}, \affordable)
     \]
     In a nutshell, at each step, a project $t \in \mathcal{T}\setminus W$ is \affordable-affordable if and only if the cost of $t$ can be covered by the voters approving $t$ in such a way that the maximum payment of any voter is \affordable. Voters who have less than a \affordable amount left spend all of their allocated budget, and the other voters pay exactly \affordable.
\end{definition}
\end{tcolorbox}

We are now ready to proceed with the guarantee for the proportionality degree of \mes. To aid our proof, we will first develop a useful lemma which upper bounds the maximum payment that can be asked from some voter in a cohesive group during the execution of the rule.
\begin{lemma}\label{lemma: maximum-payment}
    Let $T \subseteq \proj$, and let $V \subseteq N$ be a $T$-cohesive group of voters. Consider a sequence of rounds of \mes $\pi$ such that for each round $j \in \pi$, there exists some voter $i \in V$ that approves the project selected in that round, and the set of projects that are still affordable by the voters in $V$ is not empty. Then, there exists an ordering $R =\{0, \ldots, n_V-1\}$ of the voters in $V$ such that the payment of voter $i \in R$ in any of the rounds in $\pi$ is at most $\frac{\max_{t \in T}\cof{t}}{n_V-i}$. 
\end{lemma}
\begin{proof}
    To prove the lemma, we will construct an ordering by considering any round $\pi_j$ in $\pi = \{\pi_0, \ldots, \pi_h\}$ for some $h \in \mathbb{N}$. In particular, we will start with an arbitrary ordering and we will refine it in each round, making sure that it is consistent with the orderings of the previous rounds.

    Let $\pi_{\ell_1}$ be the first round in which a project $t \in \proj \setminus T$ is selected by \mes, and consider any round $\pi_j \in \{\pi_0,\ldots,\pi_{\ell_{1}-1}\}$. Let $t_j \in T$ be the $\rho$-affordable project selected by \mes in round $\pi_j$. The voters in $V$ have not paid for any project not in $T$ yet and by statement are a $T$-cohesive group, so by \cref{def: tcohesive}
    \[
    \rho \le \frac{\cof{t_j}}{n_V} \leq \frac{\max_{t\in T}\cof{t)}}{n_V} \leq \frac{\max_{t \in T}\cof{t}}{n_V-i}
    \]
    From this we obtain that up until round $\pi_{\ell_1}$, $R$ can be \emph{any} arbitrary ordering of the voters in $V$.
    
    Now consider round $\pi_{\ell_1}$, i.e., the first round that some project $t_{\ell_1}$ in $\proj \setminus T$ was selected by \mes. Let $V_{\ell_1}$ be the set of voters that were charged some positive payment from \mes for $t_{\ell_1}$, and observe that $|V_{\ell_1}| \geq 1$, since at least one voter in $V_{\ell_1}$ approves the project $t_{\ell_1}$, by the statement of the lemma. By the definition of \mes, we know that the project $t_{\ell_1}$ was $\affordable$-affordable for the smallest possible $\affordable$. This implies that $\affordable \leq \frac{\max_{t \in T}\cof{t}}{n_V}$, since there are still projects from $T$ that have not been selected: any such project $t$ is $(\cof{t}/n_V$)-affordable, since $V$ is $T$-cohesive, and all of the voters in $T$ have only paid for projects in $T$. The payment of each voter in $V_{\ell_1}$ in this round is at most $\rho$, and hence their payment is at most $\frac{\max_{t \in T}\cof{t}}{n_V-i}$. Now we will bound the payment of the voters in $V\setminus V_{\ell_1}$ for the next round. For those voters, the payment for projects in $T$ is divided equally among them, considering \cref{def: tcohesive} and that so far they have been charged solely for projects in $T$. Some voters from $V_{\ell_1}$ possibly contribute to the payment if they have some leftover budget. In the worst case, we may assume that the leftover budget of the voters in $V_{\ell_1}$ is $0$, and therefore the payment of each voter in $V\setminus V_{\ell_1}$ is at most $\frac{\max_{t\in T}\cof{t}}{n_V-|V_{\ell_1}|}$. Thus, if we refine $R$ to be $\{V_{\ell_1}, V\setminus V_{\ell_1}\}$ (with the voters in both $V_\ell$ and $V\setminus V_{\ell_1}$ ordered arbitrarily), we establish that any voter $i \in R$ pays at most $\frac{\max_{t \in T}\cof{t}}{n_V-i}$. 

    Inductively, we can further refine the ordering $R^{\ell_k}$ at any step $\pi_{\ell_k} \in \pi$, assuming that we had an ordering $R^{\ell_k-1}$ that satisfied the statement of the lemma for the previous step $\pi_{
    \ell_k-1}$. If $\pi_{\ell_k}$ is a step in which a project from $T$ is selected, then the order does not require any refinement and we can simply set $R^{\ell_k} = R^{\ell_k-1}$. If a project $t_{\ell_k} \in \proj \setminus T$ is selected at step $\pi_{\ell_k}$, then we work similarly to above. 
    
    Namely, let $V_{\overline{T}}$ be the set of all voters which contributed payments to a project in $\proj \setminus T$ before round $\pi_{\ell_k}$. Since we are considering the worst case, we have assumed that those voters have exhausted their budgets and therefore in round $\pi_{\ell_k}$ (and subsequent rounds), their payment will be $0$. For the voters in $V_{\ell_k}$, i.e., those that contributed payments to the selection of project $t_{\ell_k}$, their leftovers budgets are the same, and hence they are charged the same payment for project $t_{\ell_k}$. From the same argument as the one above, their payment of the voters in $V_{\ell_k}$ will be bounded by $\affordable \leq \frac{\max_{t \in T}\cof{t}}{n_V-|V_{\overline{T}}|}$, since there are still projects from $T$ that are affordable. Finally, again, assuming that the voters in $V_{\ell_k}$ exhaust their budget paying for project $t_{\ell_k}$, the payment of the remaining voters in $V\setminus \{V_k \cup V_{\overline{T}}\}$ is bounded by $\frac{\max_{t \in T}\cof{t}}{n_V-|V_{\overline{T}}|-|V_{\ell_k}|}$. Now we refine $R^{\ell_k}$ to be exactly as $R^{\ell_{k-1}}$, with the voters in $V_{\ell_k}$ following the voters in $V_{\overline{T}}$ and being followed by the voters in $V\setminus \{V_{\ell_k} \cup V_{\overline{T}}\}$, and the bound in the statement of the lemma follows. See \cref{fig:lemma-1} for an illustration.
\end{proof}
\begin{figure*}[h!]
    \centering
    \input{figure2}
    \caption{An illustration of the construction of the ordering $R$ in the proof of \Cref{lemma: maximum-payment}. The figure shows round $\pi_{\ell_1}$, i.e., the first round in which a project in $\proj \setminus T$ is selected, and then a subsequent round $\pi_{\ell_k}$, in which some other project in $\proj \setminus P$ is selected. $p(i)$ denotes the payment of voter $i$. By ``$any$'', it is meant that the ordering of the voters within the corresponding set is any arbitrary order.}
    \label{fig:lemma-1}
\end{figure*}
Next, we prove the bound on the proportionality degree of the rule.
\begin{theorem}\label{thm:mes-prop-degree}
    Let $T \subseteq \proj$. 
    The proportionality degree of the {\sc Method of Equal Shares} satisfies $d_{\mes}(T) \geq \frac{1}{2}\cdot\left(\frac{\cof{T}}{\max_{t \in T} \cof{t}} - 1 \right)$.
\end{theorem}
\begin{proof}
     Let $I$ be an instance of the PB problem. Let $V \subseteq N$ be any $T$-cohesive group of voters for $I$ and let $n_V = |V|$. Let $\pi = \{0, \ldots, |V| - 1\}$ be the ordering of the voters in $V$ from the statement of \cref{lemma: maximum-payment}. From the lemma, we know that the maximum payment of any voter $i \in \pi$ is at most $\frac{\max_t \cof{t}}{n_V - i}$ and furthermore for $a,b \in \pi$ such that $a<b$, if voter $b$ realizes the maximum payment in some round, then voter $a$ must already have depleted her budget at the start of that round. Since we are providing a lower bound on the average satisfaction of $V$, it suffices to impose without loss of generality the following restriction on $I$ and the execution of \mes: \medskip
    
    \noindent There exists an $i \in \pi$ such that
        \begin{itemize}[leftmargin=*]
        \item[-] for all $j \in \pi: j \leq i$, voter $j$ pays the maximum amount $\frac{\max_t \cof{t}}{n_V - j}$, and
        \item[-] for all $j \in \pi: j > i$, voter $j$ pays $0$.
    \end{itemize}
    Furthermore, when $i>0$, every voter $j \in \pi: 0 \leq j \leq i$ depletes her budget. \medskip
    
     To see why the restrictions we impose provide a lower bound on the average satisfaction, consider a total payment vector $x = (x_0, \dots, x_{n_V-1})$ induced by the execution of \mes on the rounds $\pi$. Let $F(x)=\sum_{j=0}^{n_V-1} x_j/p_j$, where $p_j$ is the maximum per round payment for which it holds by \cref{lemma: maximum-payment} that $p_j\le \max_{t\in T}\cof{t}/(n_V-j)$. Now since $F(x)$ is coordinate-wise increasing in $x$, considering $p_j$ fixed for all $j \in [n_{|V|-1}\cup \{0\}]$, setting $x_j$ to $0$ for some $j$ can only decrease $F(x)$. We can observe that the family of payment vectors $\mathcal{X}_k = (x'_0,\ldots,x'_k, \ldots, x'_{n_V-1})$ indexed by $k \in \{0,\dots,n_{|V|-1}\}$ such that $x'_i = \bud/n$ for $i \leq k-1$ and $x'_i = 0$ for $i \geq k$ provides a lower bound to $F(x)$ for any payment vector $x$. Namely, given $x$ we can choose $k$ to be at most the number of entries in $x$ such that $x_i = \bud/n$, then it is clear, by fixing $p_j$ to its maximum value for any $j$, that $F(x) \geq F(\mathcal{X}_k)$ holds \footnote{If no entry satisfies $x_i = B/n$, then we set $x'_0 = \max_i x_i$ and $x'_i = 0$ for $i > 0$.}.  
    
     Let $k \in \pi \setminus \{0\}$ denote the last voter that depleted her budget; if such a voter does not exist, then let $k=0$. We know that the voters from $V$ initially had $n_V \cdot \bud/n$ credits allocated in total. After the rule terminates, they are left with at most $\max_{t \in T} \cof{t}$ credits, since otherwise they could afford to add another project to the proposal. Therefore, the voters spent at least $n_V \cdot \bud/n - \max_{t \in T} \cof{t}$ credits and in our case the total spending, i.e. $\sum_{i=0}^{k}x_i$ is equally shared among the first $k + 1$ voters in $\pi$, since by the restriction they all deplete their budget. Thus, the average satisfaction of $V$ as a function of $k$ can be lower bounded as follows:
    \begin{align*}
        D(k) &\geq \frac{1}{n_V} \cdot \lb \sum_{i = 0}^{k}\frac{\frac{1}{k+1} \cdot \lb n_V \cdot \bud/n -    
            \max_{t \in T}\cof{t}\rb}{\frac{\max_{t \in T} 
            \cof{t}}{n_V - i}} \rb \\
            &= \lb \frac{1}{k+1} \cdot \frac{\bud/n}{\max_{t \in T} \cof{t}} - \frac{1}{(k+1)\cdot n_V}\rb \cdot \sum_{i=0}^{k}(n_V - i) \\
            &= \lb \frac{1}{k+1} \cdot \frac{\bud/n}{\max_{t \in T} \cof{t}} - \frac{1}{(k+1)\cdot n_V}\rb \cdot \lb \sum_{i=0}^{k}n_V - \sum_{i=0}^{k}i\rb \\
            &= \lb \frac{1}{k+1} \cdot \frac{\bud/n}{\max_{t \in T} \cof{t}} - \frac{1}{(k+1)\cdot n_V}\rb \cdot \lb (k+1)\cdot n_V - \frac{k \cdot (k+1)}{2} \rb \\
            &= \underbrace{\frac{\bud}{n\cdot(k+1)\cdot \max_{t \in T} \cof{t}} \cdot \lb (k+1)\cdot n_V - \frac{k\cdot(k+1)}{2}\rb}_{(\mathcal{I})} \\ 
            & \hspace{40mm}-
            \underbrace{\frac{1}{n_V \cdot (k+1)} \cdot \lb (k+1) \cdot n_V - \frac{k \cdot (k+1)}{2}\rb}_{(\mathcal{II})}
    \end{align*}
    Next, we consider the two terms separately. By expanding the first term ($\mathcal{I}$) we get:
    \begin{align*}
        &\frac{\bud}{n\cdot(k+1)\cdot \max_{t \in T} \cof{t}} \cdot \lb (k+1)\cdot n_V - \frac{k\cdot(k+1)}{2}\rb = \\  &\frac{\bud\cdot(k+1)\cdot n_V}{(k+1)\cdot n \cdot \max_{t \in T}\cof{t}}  - \frac{\bud\cdot k \cdot (k+1)}{2\cdot(k+1)\cdot n \cdot \max_{t \in T}\cof{t}} \\
        &= \frac{\bud\cdot n_V}{n \cdot \max_{t \in T}\cof{t}} -\frac{\bud\cdot k}{2\cdot n \cdot \max_{t \in T}\cof{t}}
    \end{align*}
    The second term ($\mathcal{II}$) simplifies to:
    \begin{align*}
        -\frac{1}{n_V \cdot (k+1)} \cdot \lb (k+1) \cdot n_V - \frac{k \cdot (k+1)}{2}\rb &= -\frac{(k+1)\cdot n_V}{(k+1) \cdot n_V} + \frac{k \cdot (k+1)}{2 \cdot (k+1) \cdot n_V} 
        = -1 + \frac{k}{2 \cdot n_V}
    \end{align*}
    Putting everything together, we obtain:
    \begin{equation*}
        D(k) = \frac{\bud \cdot n_V}{n \cdot \max_{t \in T} \cof{t}} - \frac{\bud \cdot k}{2 \cdot n \cdot \max_{t \in T} \cof{t}} - 1 + \frac{k}{2 \cdot n_V}
    \end{equation*}
    From the definition of cohesiveness (\Cref{def: tcohesive}), we know that $\frac{\bud}{n} \geq \frac{\cof{T}}{n_V} \geq \frac{\max_{t \in T} \cof{t}}{n_V}$. From this, it follows that $D(k)$ is a non-increasing function of $k$, and the minimum is acquired at $k = n_V - 1$. Thus, we evaluate the function at $n_V - 1$: 
    \begin{align*}
        D(n_V - 1) &= \frac{\bud \cdot n_V}{n \cdot \max_{t \in T} \cof{t}} - \frac{\bud \cdot(n_V - 1)}{2 \cdot n \cdot \max_{t \in T} \cof{t}} - 1 + \frac{n_V - 1}{2 \cdot n_V} \\
        &= \frac{\bud \cdot n_V}{n \cdot \max_{t \in T} \cof{t}} - \frac{\bud \cdot n_V}{2\cdot n \cdot \max_{t \in T} \cof{t}} + \frac{\bud}{2 \cdot n \cdot \max_{t \in T} \cof{t}} - 1 + \frac{n_V}{2 \cdot n_V} - \frac{1}{2 \cdot n_V} \\
        &= \frac{\bud \cdot n_V}{2\cdot n \cdot \max_{t \in T} \cof{t}} + \frac{\bud}{2 \cdot n \cdot \max_{t \in T} \cof{t}} - 1 + \frac{1}{2} - \frac{1}{2 \cdot n_V} \\
        &= \frac{\bud \cdot (n_V + 1)}{2\cdot n \cdot \max_{t \in T} \cof{t}} - \frac{1}{2} - \frac{1}{2 \cdot n_V} \\
        &= \frac{1}{2} \cdot \lb \frac{\bud \cdot n_V}{n \cdot \max_{t \in T} \cof{t}} - 1\rb + \frac{1}{2} \cdot \lb \frac{\bud}{n \cdot \max_{t \in T} \cof{t}} - \frac{1}{n_V}\rb \\
        &\geq \frac{1}{2}\cdot\left(\frac{\cof{T}}{\max_{t \in T} \cof{t}} - 1 \right)
    \end{align*}
    where in the last inequality we again used the definition of cohesiveness (\cref{def: tcohesive}).\medskip 
    
    \noindent This concludes the proof.
\end{proof}

Next, we provide an (almost) matching upper bound for the proportionality degree of \mes. To do that, we construct an instance with a set of projects in which the minimum average satisfaction of any $T$-cohesive group of voters is upper bounded by an appropriate quantity. Since the proportionality degree requires the guarantee to hold for every $T$, this also imposes an upper bound on the proportionality degree of \mes. 

\begin{theorem}
    There exists a PB instance $I$, a set of projects $T \subseteq \proj$, and a $T$-cohesive group $V \subseteq \VV(T)$ such that the average satisfaction of the voters in $V$ is at most $\frac{1}{2}\cdot\frac{\cof{T}}{\max_{t \in T} \cof{t}} + 1$. 
\end{theorem}
\begin{proof}
To simplify our analysis, we consider the following version of \mes with \emph{allowance}: consider some point during the execution of the rule, when a project $t$ is added to the proposal while some of the voters that approve $t$ do not have enough leftover budget to support it, and hence are compensated by the other voters in $N$. We will instead assume that (only in the first occurrence of this), the rule adds auxiliary credits to these voters ``for free'' to make the project \affordable-affordable, i.e. until the cost can be split equally among those that approve it. Then again by the definition of \mes, the rule will select the project with the minimum \affordable. It is not hard to see that the modified rule leads to a higher average satisfaction, compared to the ``vanilla'' version of \mes.

\begin{figure}[h!]
    \centering
    \input{figure1}
    \caption{The figure shows groups of voters (dots) organized into sets $Q_j$, where $j \in [0,|V|-2]$. Each group approves a distinct set $Z_j$ of projects, indicated by a blue rectangle. The set $V$ contains one voter from each $Q_j$, and each of these voters additionally approves a common set $T$ of projects (indicated as a bold blue rectangle). For $j \in \{0,\ldots,|V|-2\}$, the size of set $Q_j$ is $|V|-j$.}
    \label{fig: upper-bound}
\end{figure}
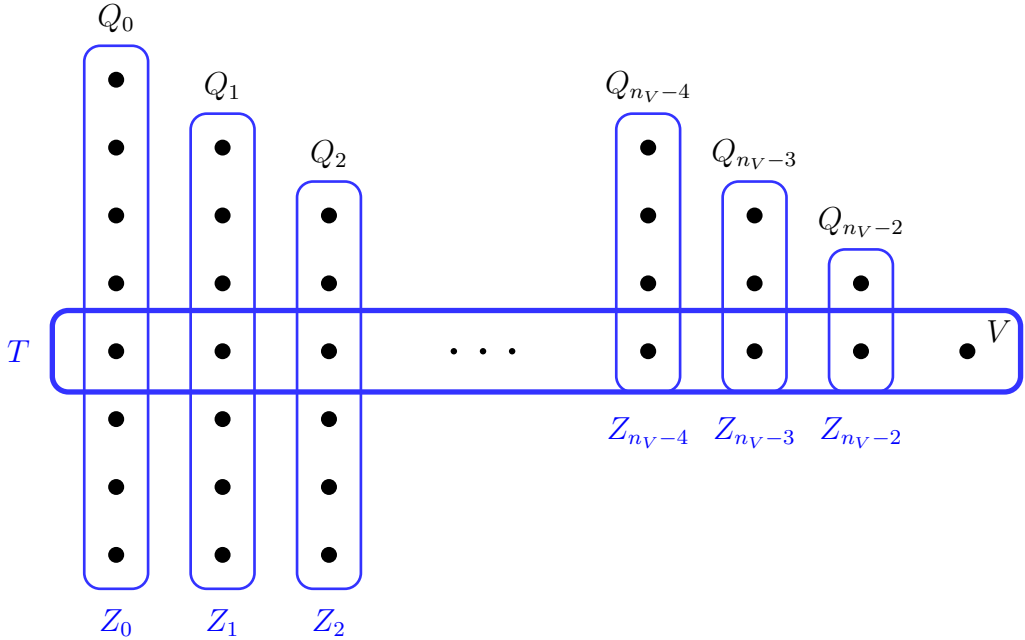

Now consider the following instance, shown in \cref{fig: upper-bound}. Let $V \subseteq N$ be a $T$-cohesive group of voters and $|V| = n_V$. For $j \in \{0, \ldots ,n_V - 2\}$, let $Q_j$ be an other set of voters, such that $|Q_j| = n_V - j$ and $|Q_j \cap V| = 1$, i.e., exactly one voter from  each one of these sets approves the projects in $T$ as well. Furthermore, the voters in the set $Q_j$ also approve a set of projects $Z_j$, for $j \in \{0, \ldots ,n_V - 2\}$. 
The cost of all projects is the same. Furthermore, we have that $|Z_j| = |T|$ and that $\frac{\bud}{n} = \frac{\cof{T}}{n_V}$. This latter fact in particular implies that the set $Q_0$ 
is $Z_0$-cohesive, but for $j \in \{1,\ldots,n_V-2\}$, the set $Q_j$ is not $Z_j$-cohesive.  

Consider the execution of \mes on the instance above. The rule selects a project that is \affordable-affordable for a minimum \affordable. In case when there are multiple candidates to be selected, we will without loss of generality assume that the rule selects a project in some set $Z_j$ before any project in $T$. In the first iteration, the candidate projects will be 
\begin{itemize}
    \item[-] the project with the minimum cost from $T$ with \affordable equals to $\frac{\min_{t \in T} \cof{t}}{n_V}$, to be split equally among the voters $V$, and 
    \item[-] the project with minimum cost from $Z_0$ with \affordable equal to $\frac{\min_{t \in Z_0} \cof{t}}{|Q_0|} = \frac{\min_{t \in T} \cof{t}}{n_V}$
\end{itemize}

By our tie-breaking rule, the project from $Z_0$ will be selected. In the second iteration, again a project from $Z_0$ will be selected. Continuing like this, all of the projects in $Z_0$ will be selected before any project from any other set will be selected, until all of the voters in $Q_0$ deplete their assigned budgets. Subsequently, for any $j \in \{1,\ldots, n_V-2\}$, the projects from $Z_j$ will be selected with \affordable equal to $\frac{\min_{t \in Z_j}\cof{t}}{n_V - j}$. In case the budget from the voters in $Q_j$ runs out, then by our assumption the rule will provide the extra option to build the project for the same \affordable, but by spending only the leftover budget.

Let $\ell$ be the single voter that approves some project in $T$ but no projects in $Q_j$ for any $j \in \{0,\ldots,n_V-2\}$. During the last iteration, all of the voters besides voter $\ell$ will have depleted their budgets to support projects in $Q_j$ for some $j \in \{0,\ldots,n_V-2\}$. Thus, voter $\ell$ will have to pay at least $\min_{t \in T} \cof{t}$ for some project in $T$. Without loss of generality, if we consider an ordering $\pi = \{0,\ldots,n_V-1\}$ of the voters in $V$, then each voter $i \in \pi$ would be charged at least $\frac{\min_{t \in T} \cof{t}}{n_V - i}$ for projects. Consider the allowance of the rule, the average satisfaction of the voters in $V$ is then at most

\begin{align*}
    \frac{1}{n_V} \cdot \lb \sum_{i = 0}^{n_V - 1} \bigg\lceil \frac{\bud/n}{ \frac{\min_{t \in T}\cof{t}}{n_V - i}}\bigg\rceil\rb &\leq \frac{1}{n_V} \cdot \lb \sum_{i = 0}^{n_V - 1} \lb \frac{\bud/n}{\frac{\min_{t \in T}\cof{t}}{n_V - i}} + 1\rb\rb \\
    &= \frac{1}{n_V} \cdot \frac{\bud/n}{\min_{t \in T}\cof{t}} \cdot \sum_{i = 0}^{n_V - 1}(n_V - i) + 1 \\
    &= \frac{1}{n_V} \cdot \frac{n_V \cdot (n_V + 1)}{2} \cdot \frac{\bud/n}{\min_{t \in T}\cof{t}} + 1 \\
    &= \frac{1}{2} \cdot \lb \frac{n_V + 1}{n_V}\rb \cdot \lb \frac{\cof{T}}{\min_{t \in T} \cof{t}}\rb + 1 \\
    &= \frac{1}{2} \cdot \lb \frac{\cof{T}}{\max_{t \in T} \cof{t}+\delta}\rb + 1,
\end{align*}
where $\delta = \max_{t \in T} \cof{t}- \min_{t \in T} \cof{t}$. By taking $n_V \rightarrow \infty$ and $\delta \rightarrow 0$, the bound follows. 

\end{proof}

\noindent Furthermore, we know that the Method of Equal Shares satisfies EJR in (approval) PB. So, similarly to the literature of multi-winner voting \citep{sanchez2017proportional} we will now provide a more general bound for every PB rule that satisfies EJR. The bound is similar to the one that we proved for \mes in \Cref{thm:mes-prop-degree}, except that we now also have a multiplicative factor of the ratio between the minimum and the maximum cost of any project in $T$. This implies that when the costs of the projects are very similar (i.e., we are close to the multi-winner setting), then the guarantee of \mes is also achieved by any other rule that satisfies EJR. Our bound will apply to instances that satisfy a natural condition, namely that for any project $t \in \mathcal{P}$, $\cof{t} \cdot n \geq B$. In other words, there is not project that is ``very cheap'', to the extend that a single voter with a $B/n$ share of the budget could financially support it single-handedly. We will refer to those instances as \emph{ordinary}.

Notice that a similar condition is imposed by \cite{sanchez2017proportional} for the multiwinner voting setting. In particular, although not explicitly mentioned, their proof (of Theorem 6) seems to require that the number of voters be at least the size of the committee; the equivalent condition in the PB setting is the one that we impose here. 

\begin{theorem}
    Let $T \subseteq \proj$ and let $\pbrule$ be a PB rule that satisfies EJR. Then, the proportionality degree of $\pbrule$ on any ordinary instance satisfies \[
    d_{\pbrule}(T) \geq \frac{\min_{t \in T}\cof{t}}{\max_{t \in T}\cof{t}} \cdot \frac{1}{2} \cdot \left(\frac{\cof{T}}{\max_{t \in T} \cof{t}} - 1\right)
    \] 
\end{theorem}
\begin{proof}
Let $s = \frac{n}{\bud}$, and let $|V| = n_V$. For the proof, we will construct a set of groups of voters $N_j$ for $j\in \{m,m-1,\ldots, 1\}$. We will also fix an order of $\pi$ of the projects in $T$, namely $\{t_1,t_2,\ldots,t_m\}$. We will prove the following properties:
\begin{enumerate}
    \item for every $j$, for every voter $i \in N_{m-j}$, it holds that $|A_i \cap W| \geq m-j$,\label{enum-here:item1}
    \item for every $j$, $|N_{m-j}| \geq \min_{t \in T} \cof{t} \cdot s$,\label{enum-here:item2}
    \item $|V\setminus N_m| < \cof{T} \cdot s$.\label{enum-here:item4}
\end{enumerate}
We will first explain how to construct group $N_m$. At the beginning, we set $N_m = \emptyset$. At that point, we have that $|V \setminus N_m| \geq \cof{T}\cdot s$; this follows from the fact that $V$ is $T$-cohesive. Next, we iteratively add voters from $V$ to $N_m$, as long as the following inequality still holds:
\begin{equation}
|V\setminus N_m| \geq \cof{T} \cdot s.
\end{equation}
In particular, after the last voter $i_\ell$ is added to $N_m$, we have
\begin{equation}\label{eq:VsetminusN_m}
|V\setminus N_m| < \cof{T} \cdot s,
\end{equation}
therefore obtaining \Cref{enum-here:item4} above. Notice that the set $V\setminus N_m \cup i_\ell$ (i.e., the set $V\setminus N_m$ before the last voter $i_\ell$ was added to $N_m$) is $T$-cohesive. Additionally, each of the voters in $N_m$ approves all projects in $T$, since $T \subseteq \bigcup_{i \in V} A_i$. As a result, by the EJR property (\Cref{def: EJR-PB}), we obtain that for every voter $i \in N_m$, it holds that $|A_i \cap W| \geq m$. We have thus obtained \Cref{enum-here:item1,enum-here:item2} above for the set $N_m$.

Next, we explain how to construct the group $N_{m-1}$; the construction of the remaining groups for $j =\{m-2,\ldots,1\}$ will be identical. Similarly to above, we start with $N_{m-1} = \emptyset$ and we add voters from $V\setminus N_m$, as long as the following inequality holds:
\begin{equation}
|V\setminus N_m \setminus N_{m-1}| \geq \cof{T\setminus{\{t_1\}}} \cdot s.
\end{equation}
Again, after the last voter $i_\ell'$ is added to $N_{m-1}$, we have 
\begin{equation}
|V\setminus N_m \setminus N_{m-1}| < \cof{T\setminus{\{t_1\}}} \cdot s.
\end{equation}
Similarly to above, the set $V\setminus N_m \setminus N_{m-1} \cup i_\ell'$ is $(T \setminus \{t_1\})$-cohesive, and each of the voters in $N_{m-1}$ approves all projects in $T\setminus \{t_1\}$. So, by the EJR property, we obtain that for every voter $i \in N_{m-1}$, it holds that $|A_i \cap W| \geq m-1$. We have thus obtained \Cref{enum-here:item1} above for the set $N_{m-1}$. To obtain \Cref{enum-here:item2}, observe that $|N_m| = \lfloor n_V - \cof{T}\cdot s\rfloor + 1$, and hence $|V \setminus N_m| = n_V - |N_m| \geq \cof{T} \cdot s -1$. Therefore:
\begin{align*}
|V \setminus N_m| & \geq \cof{T} \cdot s -1 \\
&= \cof{T} \cdot s -1 + \cof{t_1} \cdot s - \cof{t_1}\cdot s \\
&= \cof{T\setminus\{t_1\}}\cdot s + \cof{t_1}\cdot s -1
\end{align*}
Since the instance is ordinary, we have that $\cof{t_1}\cdot s \geq 1$, it also holds that $|V \setminus N_m|\geq \cof{T\setminus\{t_1\}}\cdot s$. This implies that the size of $|N_{m-1}| \geq \cof{t_1}\cdot s \geq \min_{t \in T} \cof{t} \cdot s$, thus obtaining \Cref{enum-here:item2}. 

Now that we have \Cref{enum-here:item1,enum-here:item2,enum-here:item4}, we are ready to prove the bound on the proportionality degree. The average satisfaction of voters in $V$ can be written as a convex combination of the average satisfaction of the voters in $N_m$ and those in $V \setminus N_m$. Since the average satisfaction of the former set of voters is $m$ (from \Cref{enum-here:item1}), we will lower bound the average satisfaction of the voters in $V$, by that of the voters in $V \setminus N_m$. This latter quantity can be expressed as:

\begin{align*}
    \frac{1}{|V \setminus N_m|} \cdot \sum_{i \in V \setminus N_m}|A_i \cap W| 
    & \geq \frac{1}{\cof{T} \cdot s} \cdot \sum_{i \in V \setminus N_m}|A_i \cap W| \\
    &\geq \frac{1}{\cof{T} \cdot s } \cdot \sum_{k=1}^{m-1}|N_k|\cdot k \\ 
    &\geq \frac{1}{\cof{T} \cdot s} \cdot s \cdot \min_{t \in T}\cof{t} \cdot \sum_{k=1}^{m-1}k \\
    &= \frac{m \cdot \min_{t \in T}\cof{t}}{\cof{T}} \cdot \left(\frac{m - 1}{2}\right) \\
    &\geq \frac{\min_{t \in T}\cof{t}}{\max_{t \in T}\cof{t}} \cdot \frac{1}{2} \cdot \left(\frac{\cof{T}}{\max_{t \in T} \cof{t}} - 1\right)
\end{align*}
where the first inequality follows from \Cref{enum-here:item4}, the second inequality follows from \Cref{enum-here:item1}, and the third inequality follows from \Cref{enum-here:item2}. This concludes the proof. 
 
\end{proof}

\newpage
\section{Proportionality Degree of the Phragmén Sequential Rule}
We now move on to studying the proportionality degree of another well-known rule for PB, namely the {\sc Phragmén} Sequential rule \citep{Phragmen1894methode}. This rule does not achieve the strong proportionality properties of EJR, but it satisfies the weaker axiom of Proportional Justified Representation (PJR), see \cite{brill2024Phragmen}. One might feel inclined to believe that the stronger axiomatic properties of \mes would result in a larger proportionality degree as well. However, as observed also by  \cite{skowron2021proportionality}, it is clearly conceivable for a rule to violate EJR while achieving a high proportionality degree. This is because while the satisfaction of \emph{one} voter may be low, the satisfaction of the other voters in the same cohesive group may be fairly high. 

We provide the definition of the {\sc Phragmén} Sequential rule for PB below. Our definition follows the recent literature on computational social choice, e.g., see \citep{los2022proportional,peters2021proportional}. 

\begin{tcolorbox}[
    standard jigsaw,
    opacityback=0,
]
    \begin{definition}[{\sc Phragmén} rule \citep{los2022proportional}]
    The rule starts with an empty proposal. Every voter gets credits continuously at the rate of one unit of credits per unit of time. At the first moment $t$ when there is a group of voters $V$ who all approve a not-yet-selected project $p$, and who together have $\cof{p}$ units of credits, the rule adds $p$ to the proposal and asks the voters from $V$ to pay the cost of $p$ (i.e. the rule resets the balance of each voter from $V$), while the others keep their so far earned money. The rule stops when it would select a project which would overshoot the budget.     
\end{definition}
\end{tcolorbox}

Next, we present the proofs for the proportionality degree of the {\sc Phragmén} rule for the approval-based PB model. Our lower bound proof leverages a potential function argument, similarly as in \cite{skowron2021proportionality} for the case of multiwinner elections. Namely we have the following theorem.

\begin{theorem}
    Let $T \subseteq \proj$, for each $T$-cohesive group of voters the proportionality degree of the {\sc Phragmén} rule satisfies $d_{\textit{Phrag}}(T) \geq \frac{1}{2} \cdot \left(\frac{\cof{T}}{\max_{t \in T} \cof{t}} - 1 \right)$
\end{theorem}
\begin{proof}
    Given a participatory budgeting instance $I$, let $W$ be the set of projects selected by the {\sc Phragmén} rule. For the sake of contradiction, assume that there exists a $T$-cohesive group of voters $V$ with $|T| \geq \frac{1}{2} \cdot \left(\frac{\cof{T}}{\max_{t \in T} \cof{t}} - 1 \right)$ and  average satisfaction $\text{avg}_W(V) < \frac{1}{2} \cdot \left(\frac{\cof{T}}{\max_{t \in T} \cof{t}} - 1 \right)$. To prove the desired bound, we consider the ratio between the total amount of credits elicited from voters in $V$ during the execution of the {\sc Phragmén} rule over an upper bound on the average individual spending of a voter from $V$. We calculate separately these two quantities:
    \begin{itemize}
        \item[-] \textit{Total credits spent by $V$}: Consider the time $\tau_\ell$ when the rule terminates; this is the time where a project $t_\ell$ was about to be selected, but its selection would overshoot the budget. Let $c$ be the amount of virtual credits earned by all voters up to time $\tau_\ell$. We have:
        \begin{equation}\label{ineq: total-credits}
            c = \tau_\ell \cdot n = \cof{W} + \cof{t_\ell} + x > B
        \end{equation}
        where $x \geq 0$ is the amount of credits that non-supporters of $t_\ell$ (i.e., any voter $i$ such that $t_\ell \notin A_i$) earned in the meantime. The voters in $V$ together have earned $\frac{|V|}{n}\cdot c$ credits and since $V$ is a $T$-cohesive group, it holds that
        \begin{equation}
            \frac{|V|}{n}\cdot c \geq \frac{c \cdot \cof{T}}{\bud} \geq \cof{T},
        \end{equation}
        where the second inequality holds from (\ref{ineq: total-credits}). By assumption, the amount of credits left to $V$ after all the projects are selected is at most $\max_{t \in T} \cof{t}$. Indeed, assume otherwise; then these credits would have been spent earlier for selecting an additional project approved by all the voters from $V$, and such a project would exist by assumption. Thus the voters from $V$ will spent at least $\cof{T} - \max_{t \in T}\cof{t}$ for supporting the projects that they approve.
      
        \item[-] \textit{Average spending by voters in $V$}: Similarly to \citep{skowron2021proportionality} for a time $\tau$ we define the potential value $\phi_\tau$ as:
        \[
        \phi_\tau = \sum_{i \in V}(p_\tau(i) - p_\tau(V))^2
        \]
        where $p_t(i)$ denotes the number of credits held by voter $i \in V$ at time $\tau$ and $p_\tau(V) = \frac{1}{|V|}\sum_{i \in V}p_\tau(i)$ denotes the average amount of credits of the voters in $V$. Now let $\Delta_{\phi}$ be the change in the potential value when a voter $j \in V$ pays for a project, i.e., when the amount of her credits becomes $0$. From \citep{skowron2021proportionality} we know that
        \begin{equation}\label{eq: potential-change}
            \Delta_{\phi} = p_\tau(j)\left(2p_\tau(V) - \frac{p_\tau(j)(|V| + 1)}{|V|}\right)
        \end{equation}
        Furthermore, in each time $\tau$ we have that $p_\tau(V) \leq \frac{\max_{t \in T}\cof{t}}{|V|}$ as otherwise the sum of credits within the group would be greater than $\max_{t \in T}\cof{t}$, and such credits would be earlier spent on buying a project who is approved by all the voters within the group. 

        Next let $x_{\tau,j} = p_\tau(j) - \frac{2|V|}{|V| + 1} \frac{\max_{t \in T}\cof{t}}{|V|}$. Then from equation (\ref{eq: potential-change}) we have:
        \[
        \Delta_{\phi} = p_\tau(j) \left(2p_t(V) - 2\cdot \frac{\max_{t \in T}\cof{t}}{|V|} - \frac{x_{\tau,j}(|V| + 1)}{|V|}\cdot p_\tau(j) \right)
        \]
        If $x_{\tau,j} > 0$, then $\phi$ decreases by at least $2 |x_{\tau,j}| \cdot \frac{\max_{t \in T}\cof{t}}{|V|}$. Similarly, if $x_{\tau,j} \leq 0$, then $\phi$ increases by at most $2 |x_{\tau,j}| \cdot \frac{\max_{t \in T}\cof{t}}{|V|}$. Since the potential value is always non-negative, we infer that the values of $x_{\tau,j}$ for $j \in V$ and $\tau$ are on average at most $0$. Thus, the voters from $V$ on average pay at most $\frac{2|V|}{|V| + 1}\cdot \frac{\max_{t \in T}\cof{t}}{|V|} < \frac{2\max_{t \in T}\cof{t}}{|V|}$ for a project.
    \end{itemize}
    Consequently, the average number of projects that the voters approve equals at least:
    \[
    \frac{\cof{T} - \max_{t \in T}\cof{t}}{\frac{2\max_{t \in T}\cof{t}}{|V|} \cdot |V|} = \frac{\cof{T} - \max_{t \in T}\cof{t}}{2\cdot \max_{t \in T}\cof{t}}
    \]
This concludes the proof. 
\end{proof}

Next we provide a matching up to an additive constant upper bound for the proportionality degree of Phragmén Sequential Rule.

\begin{theorem}
    Let $T \subseteq \proj$, for each $T$-cohesive group of voters the proportionality degree of the {\sc Phragmén} rule satisfies $d_{\textit{Phrag}}(T) \leq \frac{1}{2} \cdot \left(\frac{\cof{T}}{\max_{t \in T} \cof{t}}  \right)$
\end{theorem}
\begin{proof}
We construct the following instance on which we will upper bound the minimum average satisfaction over $T$-cohesive groups. We will construct three disjoint sets of projects, namely:
\[
  Z_1 = \{b_1, \ldots, b_{m_1}\}, \ \ \ \ 
  Z_2 = \{c_1, \ldots, c_{m_2}\}, \ \ \ \ 
  T = \{d_1, \ldots, d_{m_T}\}, \ \ \ \ 
\]
and therefore $\proj = Z_1 \cup Z_2 \cup T$. Also, $\cof{T}$ will be such that $\cof{T} < B/2$ and $B$ is divisible by $\cof{T}$, i.e., there exists an $\alpha \in \mathbb{Z}$ such that $B = \alpha \cdot \cof{T}$. Also, let 
$$x = \frac{2 \cdot (\bud -\cof{T})}{\cof{T}} - 1$$ 
and let $R \in \mathbb{Z}$ divisible by $x$, i.e., there exists  $\beta \in \mathbb{Z}$, such that $R = \beta \cdot x$. Note that $x$ is an integer, since it holds that $x = 2\alpha - 2$. For every project $t' \in Z_1 \cup Z_2$, we have $cost(t') = \min_{t \in T}\cof{T}$.

Let the number of voters be $n = \frac{R\cdot \bud}{\cof{T}}$. The number of voters $n$ will also be equal to $\min_{t}\cof{T}$ (and hence $\min_{t\in T}\cof{T} \in \mathbb{Z}$). Note that this possible: for every $\alpha \in \mathbb{Z}$, we can choose $\beta$ such that $(2\alpha \beta - 2\beta)\alpha = \min_{t \in T}\cof{T}$; we leave $\alpha$ to be chosen later in the proof.

In the following, we define the approval preferences of the voters. For convenience, instead of defining the approval set $A_i$ of every voter $i \in N$, we will define $A(t)$ to be the set of voters that approve project $t$, and we will define $A(t)$ for each $t \in \proj$. First, let 
\[
C_R=\{R+1, \ldots, \min_{t \in T}\cof{t} - R\}
\]
First, for $j \in [1,x]$, define
\[
A(b_j) = \{(j-1)\cdot R/x+1,(j-1)\cdot R/x+2,\ldots,j\cdot R/x\} \cup \{R+1,R+2,\ldots \min_{t\in T} \cof{T} - j\cdot R/x\}.
\]
Then, for $j \in [x+1,\ldots, m_1\}$, we define
\begin{align*}
&A(b_{x+1}) = \{1,2,\ldots, R/x\} \cup C_R \\
&A(b_{x+2}) = \{R/x+1, \ldots, 2R/x\} \cup C_R \\
& \vdots \\
& A(b_{m_1}) = (\text{a ``cyclic shift'' of} \{1,2,\cdots, R/x\}) \cup C_R\}
\end{align*}
Then, for each $j \in \{1,\ldots,m_2\}$, let 
\[
A(c_j) = \{\min_{t\in T}\cof{t}-R, \min_{t \in T}\cof{T}-R+1, \ldots, \min_{t \in T}\cof{T}\}
\]
Then, for each $j \in \{1,\ldots,m_3\}$, let 
\[A(d_j) = \{1, \ldots, R\}.\]
    
Thus, 
\begin{itemize}
    \item[-] project $b_1$ is approved by $\min_{t \in T}\cof{t} - R$ voters in total, 
    \item[-] project $b_2$ is approved by $\min_{t \in T}\cof{t} - R - R/x$ voters in total,  
    \item[-] project $b_3$ by $\min_{t \in T}\cof{t} - R - 2R/x$ voters in total, 
    \item[-] each project from $\{b_x, \ldots, b_k\}$ is approved by the same $\min_{t \in T}\cof{t} - 2R$ voters from $\{R+1, \ldots, \min_{t \in T}\cof{t} - R\}$ and by some $R/x$ voters from $\{1, \ldots, R\}$, which cyclically shift, 
\end{itemize} 
The voters from $\{1, \ldots, R\}$ who approve $b_i$, for $i \geq x$, are those who are right after the voters who approve the $(b_{i-1})$-th project, unless those who approve $b_{i-1}$ from the last segment are $\{R - R/x + 1, \ldots, R\}$. In such a case, the voters from $\{1, \ldots, R\}$ who approve $b_i$ are exactly $\{1, 2, \ldots, R/x\}$. Furthermore, we can see that the group of voters $\{1, \ldots, R\}$ form a $T$-cohesive group.\medskip

\noindent Now, consider all the voters approving $b_1$. Let $\tau = \frac{\bud}{\bud -\cof{T}}$; this ensures that
\[\tau \cdot \left(\min_{t\in T}\cof{t} - R\right) = \frac{\bud}{\bud -\cof{T}} \cdot \left(\min_{t \in T}\cof{t} - \min_{t \in T}\cof{t} \cdot \frac{\cof{T}}{\bud} \right) = \min_{t \in T}\cof{t}.\]

Thus, \textsc{Phragmén’s} Sequential Rule selects $b_1$ first at time $\tau$ \footnote{We assume adversarial tie-breaking against the projects in $T$.}. Indeed, at this time the group of voters $A(b_1)$ collects $\tau \cdot (\min_{t \in T}\cof{t} - R) = \min_{t \in T}\cof{t}$ credits. At time $2\tau$, each voter from $\{R/x + 1, \ldots, 2R/x\}$ has already $2\tau$ credits, and each voter from $\{R + 1, R + 2, \ldots, \min_{t \in T}\cof{t} - 2R/x\}$ has $\tau$ credits; altogether, they have $\min_{t \in T}\cof{t}$ credits, thus $b_2$ is selected second. By a similar reasoning, we infer that in the first $x$ steps candidates $b_1, \ldots, b_x$ will be selected by \textsc{Phragm\'{e}n’s} Sequential Rule, and the last one of them will be selected at time $x\cdot \tau$.

At this time the rule would also select one candidate from $\{c_1, \ldots, c_{m_2}\}$. Indeed, at time $x\cdot \tau$ the voters from $\min_{t \in T}\cof{t} - R, \min_{t \in T}\cof{t} - R + 1, \ldots, \min_{t \in T}\cof{t}$ have the following number of credits:
\begin{align*}
    \frac{R}{x} \cdot \tau \cdot (1 + 2 + \ldots + x) &= \frac{R}{x} \cdot \tau \cdot \frac{x(x + 1)}{2} = \frac{R \cdot \tau \cdot (x + 1)}{2} \\
    &= \frac{n \cdot \cof{T} \cdot \tau \cdot (x + 1)}{2 B} = \frac{\min_{t\in T}\cof{t} \cdot \cof{T} \cdot \tau \cdot (x + 1)}{2 B} \\
    &= \min_{t\in T}\cof{t}
\end{align*}
In the last step we substitute $\tau$ and $x$ and the terms cancel out accordingly.

Let us now analyze what happens before the budget $\bud$ is depleted. First, let us consider how the \textsc{Phragm\'{e}n}’s Sequential Rule would behave if there were no candidates from $\{c_1, \ldots, c_{m_2}\}$. At time $(x+1)\cdot \tau$ voters from $\{1, 2, \ldots, R/x\}$ have $x\cdot \tau$ credits each. Similarly, each voter from $\{R/x+1, \ldots, 2R/x\}$ has $(x-1)\tau$ credits, each voter from $\{2R/x + 1, \ldots, 3R/x\}$ has $(x - 2)\tau$ credits, etc. The amount of credits held by the voters from $\{1, 2, \ldots, R\}$ altogether at time $(x + 1)\tau$ is:
\[
\frac{R}{x} \cdot \tau \cdot (1 + 2 + \ldots + x) = \min_{t \in T}\cof{t}
\]

At the same time voters from $A(b_{x+1})$ have $\min_{t \in T}\cof{t}$ credits altogether, thus, $b_{x+1}$ can be selected next, before any project from $D = \{d_1, \ldots, d_{m_T}\}$ is chosen, since we assumed adversarial tie-breaking. Similarly, at time $(x + 2)\tau$ each voter from $\{1, 2, \ldots, R/x\}$ has $\tau$ credits, each from $\{R/x + 1, \ldots, 2R/x\}$ has $x\cdot \tau$ credits, each from $\{2R/x + 1, \ldots, 3R/x\}$ has $(x - 1)\tau$ credits, etc.; altogether they have at most $\min_{t \in T}\cof{t}$ credits, and the voters from $A(b_{x+2})$ have exactly $\min_{t \in T}\cof{t}$; thus $b_{x+2}$ can be selected next. Through a similar reasoning we conclude that before the budget $\bud$ is depleted, \textsc{Phragm\'{e}n}’s Sequential Rule would select projects $\{b_1, \ldots, b_{m_1}\}$.

Now, consider the projects from $\{c_1, \ldots, c_{m_2}\}$. After projects $b_1, \ldots, b_x$ are selected, projects from $\{c_1, \ldots, c_{m_2}$ are approved only by voters who do not approve other remaining projects, except the voter indexed by $\min_{t \in T}\cof{t} - R$ who approves all the projects $b_x, \ldots, b_{m_1}$. Thus, their selection does not interfere with the relative order of selecting the other projects. Further, observe that over time $(x + 1)\cdot \tau$ the last $R$ voters collect the following number of credits:
\[
(x + 1)\cdot \tau \cdot R = \frac{2(\bud -\cof{T})}{\cof{T}} \cdot \frac{\bud}{\bud -\cof{t}} \cdot \min_{t \in T}\cof{t} \cdot \frac{\cof{T}}{\bud} = 2\cdot \min_{t \in T}\cof{t}
\]

Thus, every at every $\left(\frac{(x + 1)}{2}\tau\right)$-th time step, the rule will select one project from $\{c_1, \ldots, c_{m_2}\}$. Next we show that in the first 
$$\tau\cdot \left(\frac{\bud}{\min_{t \in T}\cof{t}} - \frac{\frac{2\bud}{\min_{t \in T}\cof{t}} - 2x}{x + 3}\right)$$ time steps, the rule will select at least $\left\lfloor \frac{\frac{2\bud}{\min_{t \in T}\cof{t}} - 2x}{x + 3} \right\rfloor + 1$ projects from $\{c_1, \ldots, c_{m_2}\}$. \medskip 

\noindent Indeed, the first project will be selected after the first $x\cdot \tau$ time steps. After the remaining 
$$\tau \cdot \left(\frac{\bud}{\min_{t \in T}\cof{t}} - x - \left\lfloor \frac{\frac{2\bud}{\min_{t \in T} \cof{t}} - 2x}{x + 3} \right\rfloor\right)$$ 
time steps the number of projects from $\{c_1, \ldots, c_{m_2}\}$ that will be selected is equal to:
\begin{align*}
    \frac{\tau\cdot \left(\frac{\bud}{\min_{t \in T}\cof{t}} - x - \left\lfloor \frac{\frac{2\bud}{\min_{t \in T}\cof{t}} - 2x}{x + 3} \right\rfloor\right)}{\frac{(x+1)t}{2}} &\geq \frac{2\cdot\left(\frac{\bud}{\min_{t \in T}\cof{t}} - x - \frac{\frac{2\bud}{\min_{t \in T}\cof{t}} - 2x}{x + 3}\right)}{x + 1} \\
    &= \frac{\frac{2\bud}{\min_{t \in T}\cof{t}} - 2x}{x + 3} > \left\lfloor \frac{\frac{2\bud}{\min_{t \in T}\cof{t}} - 2x}{x + 3} \right\rfloor
\end{align*}

Next, observe that:
\begin{align*}
    \frac{\frac{2\bud}{\min_{t \in T}\cof{t}} - 2x}{x + 3} + 1 &\geq \left\lfloor \frac{\frac{2\bud}{\min_{t \in T}\cof{t}} - \frac{4\bud -6\cof{T}}{\cof{T}}}{\frac{2\bud -3\cof{T}}{\cof{T}} + 3} \right\rfloor + 1 \\
    &= \left\lfloor \frac{\frac{2\bud}{\min_{t \in T}\cof{t}}\cdot \cof{T} - 4B + 6\cdot \cof{T}}{2\bud} \right\rfloor + 1 \\
    &\geq \frac{\cof{T}}{\min_{t \in T}\cof{t}} - 1 \geq |T| -1
\end{align*}

Consequently, the budget spent on projects from $\{b_1, \ldots, b_{m_{T}}\}$ is at most $\bud - (|T| - 1)\cdot \min_{t \in T}\cof{t} $ and the rest of the budget is spent on projects from $\{c_1, \ldots, c_{m_2}\}$.

The group of voters $V = \{1, \ldots, R\}$ is $T$-cohesive; we will next assess their average number of representatives. Observe that, except for projects from $\{c_1, \ldots , c_{m_2}\}$, each project from the selected proposal is approved by exactly $\frac{R}{x}$ voters of $V$. Thus:
\begin{align*}
    \frac{1}{|V|} \sum_{i \in V} |W \cap A_i| &= \frac{1}{R} \cdot \frac{R}{x} \cdot \left(\frac{\bud -(|T| - 1)\cdot \min_{t \in T}\cof{t}}{\min_{t \in T}\cof{t}}\right) \\ 
    &= \frac{\cof{T}}{2\cdot \min_{t \in T}\cof{t}} \cdot \frac{2\bud -2(|T| - 1)\cdot \min_{t \in T}\cof{t}}{2\bud -3\cdot \cof{T}}
\end{align*}
Next, choose $|Z_1|= |Z_2| = \lfloor B\rfloor$, for $\bud$ large enough; by choosing $\alpha$ accordingly, the second term is arbitrarily close $1$. Also, choose $\min_{t \in T}\cof{t}$ such that $\frac{\max_{t \in T}\cof{t}}{\min_{t\in T}\cof{t}} \leq 1 + \delta$ with $\delta \rightarrow 0^{+}$ and we obtain the desired bound.
\end{proof}

\newpage

\section{Experimental Evaluation}
To complement our theoretical analysis, we perform experiments using real-world data from the Pabulib library of \citet{faliszewski2023participatory}. The library contains instances of participatory budgeting with different numbers of voters and projects, costs, as well as voters' approvals for the projects. We select 100 instances from the library and evaluate five different PB rules on those instances, namely (i) \mes (ii) Sequential Phragm\'{e}n (iii) \mes with exhaustion (iv) Phragm\'{e}n with exhaustion and (v) greedy approval. The greedy approval rule \citep{talmon2019framework} sorts the projects according to their cumulative approval scores and allocates the budget according to this order until no more projects can be funded. By ``with exhaustion'', we mean versions of the rules which, at the end of their execution, assuming there is leftover budget and projects to be selected, select the remaining projects using the greedy approval rule. 

Due to its simplicity, the greedy rule is widely used in practice \cite{peters2021proportional,laruelle2021voting}. Given this, it is justified to study the performance of the greedy rule and compare it with that of \mes and Phragm\'{e}n. Next we present  our experimental methodology and results. 

\begin{figure}[h!]
    \centering 
\begin{subfigure}{0.4\textwidth}
  \includegraphics[width=\linewidth]{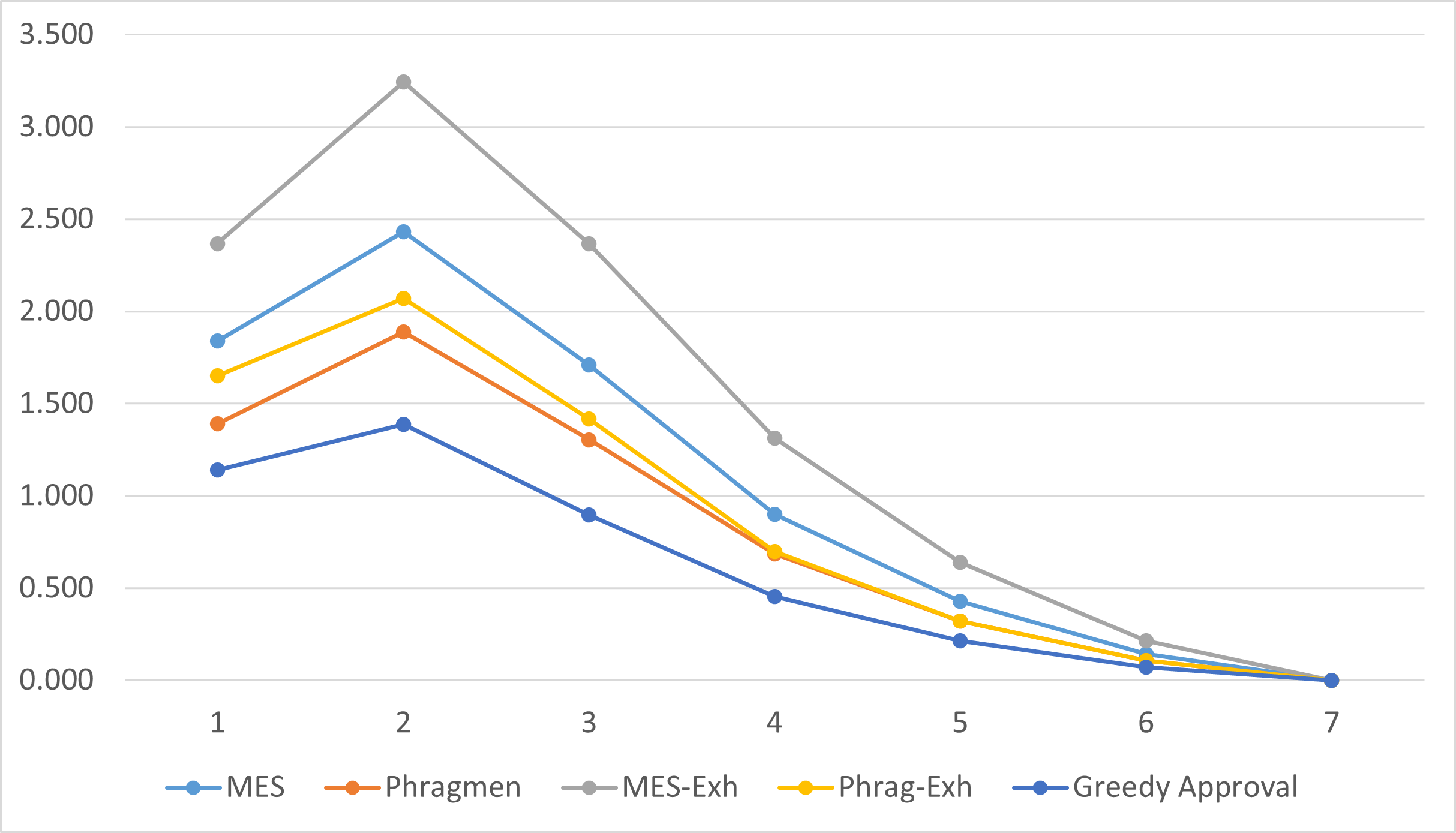}
  \label{fig:1}
\end{subfigure}\hfil 
\begin{subfigure}{0.4\textwidth}
  \includegraphics[width=\linewidth]{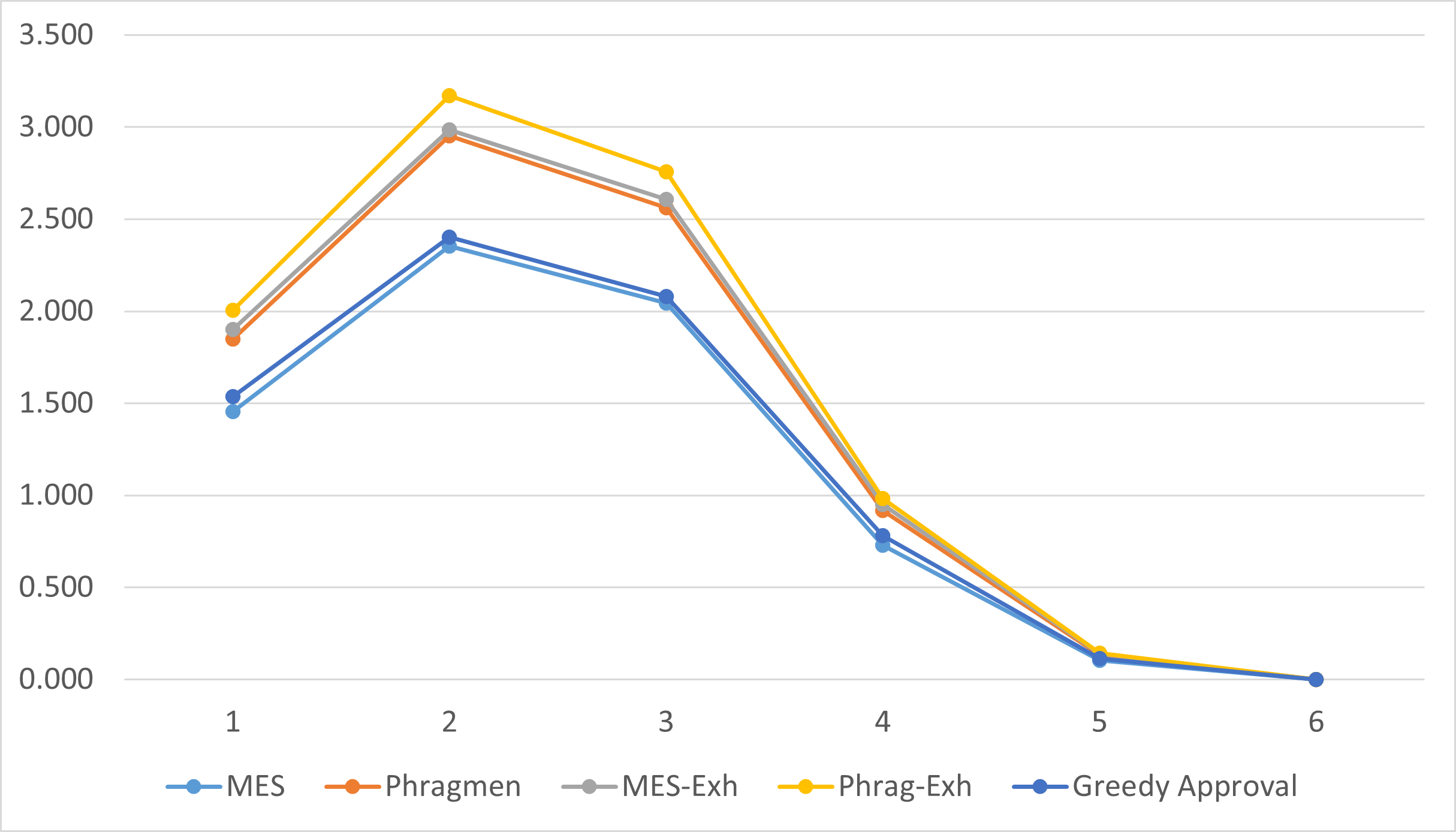}
  \label{fig:2}
\end{subfigure}\hfil 
\begin{subfigure}{0.4\textwidth}
  \includegraphics[width=\linewidth]{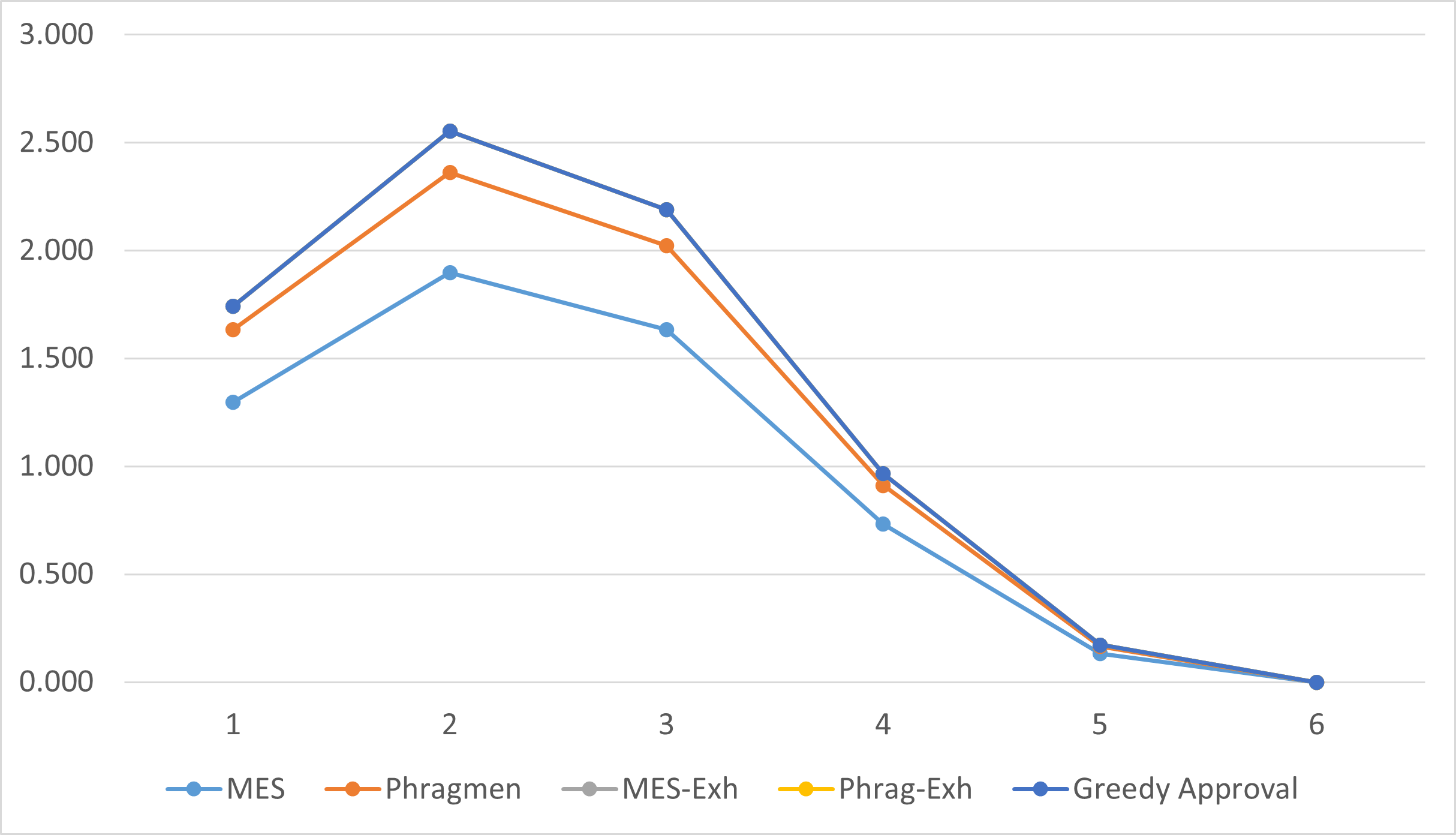}
  \label{fig:3}
\end{subfigure}
\bigskip

\begin{subfigure}{0.4\textwidth}
  \includegraphics[width=\linewidth]{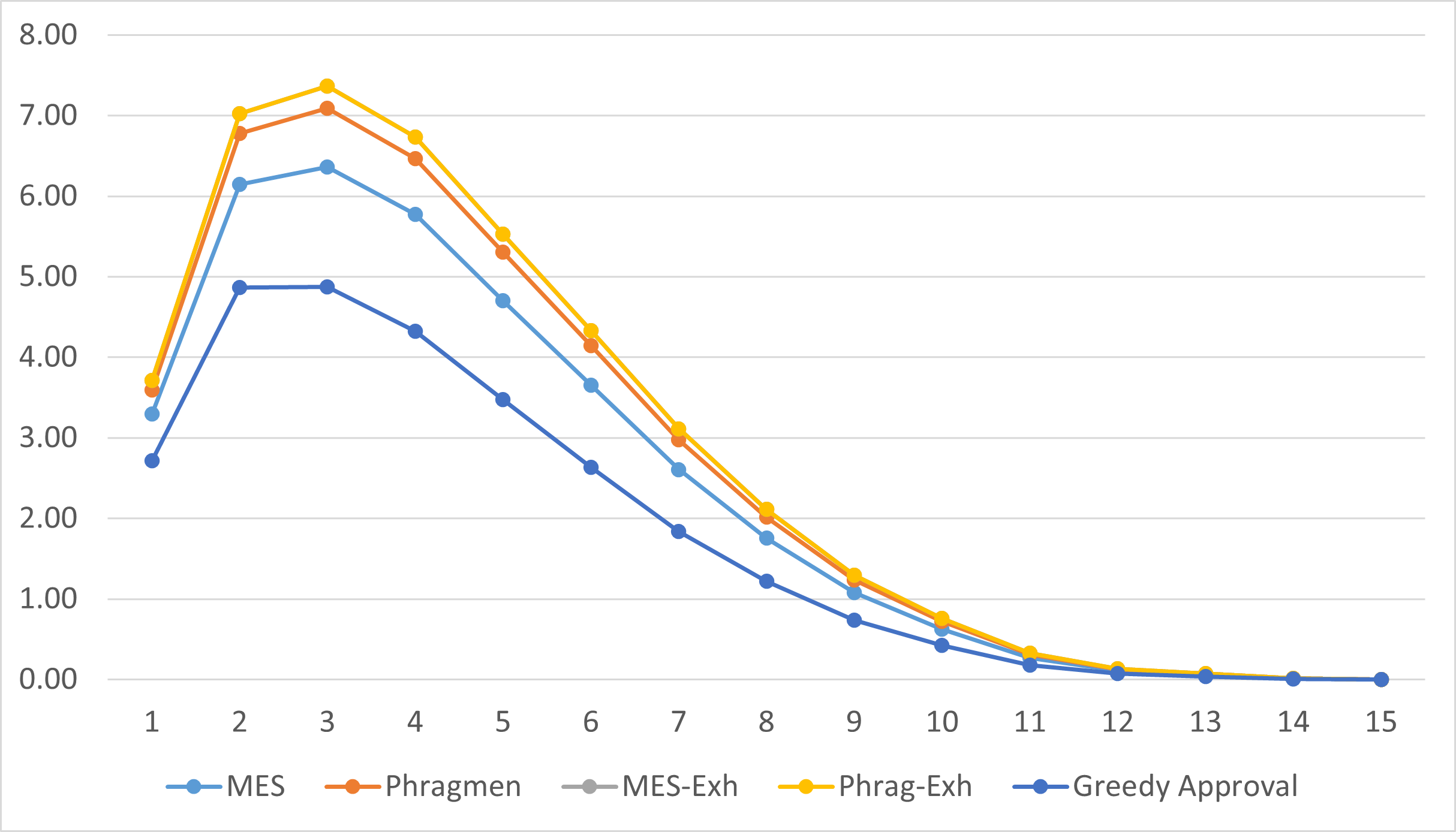}
  \label{fig:4}
\end{subfigure}\hfil 
\begin{subfigure}{0.4\textwidth}
  \includegraphics[width=\linewidth]{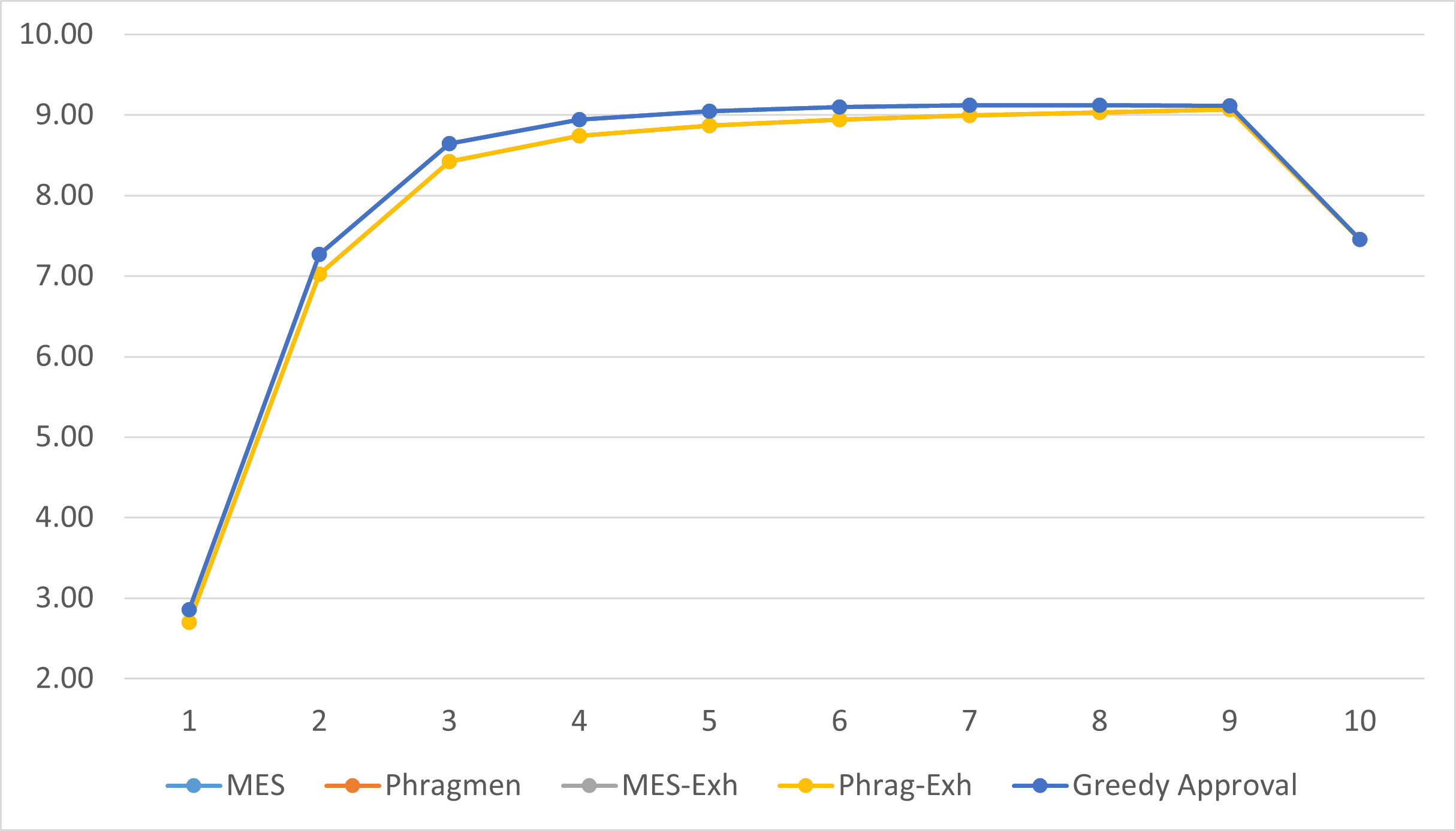}
  \label{fig:5}
\end{subfigure}\hfil 
\caption{The graphs show the results for 5 different representative datasets out of the 100. Note that the suffix `Exh' represents the rule with budget exhaustion. For each dataset, the $X$-axis represents the subset size, while the $Y$-axis represents the average proportionality degree of all the sampled subsets of that size.}
    \label{fig: all_datasets}
\end{figure}

\begin{figure}
    \centering
    \begin{subfigure}{0.5\textwidth}
  \includegraphics[width=\linewidth]{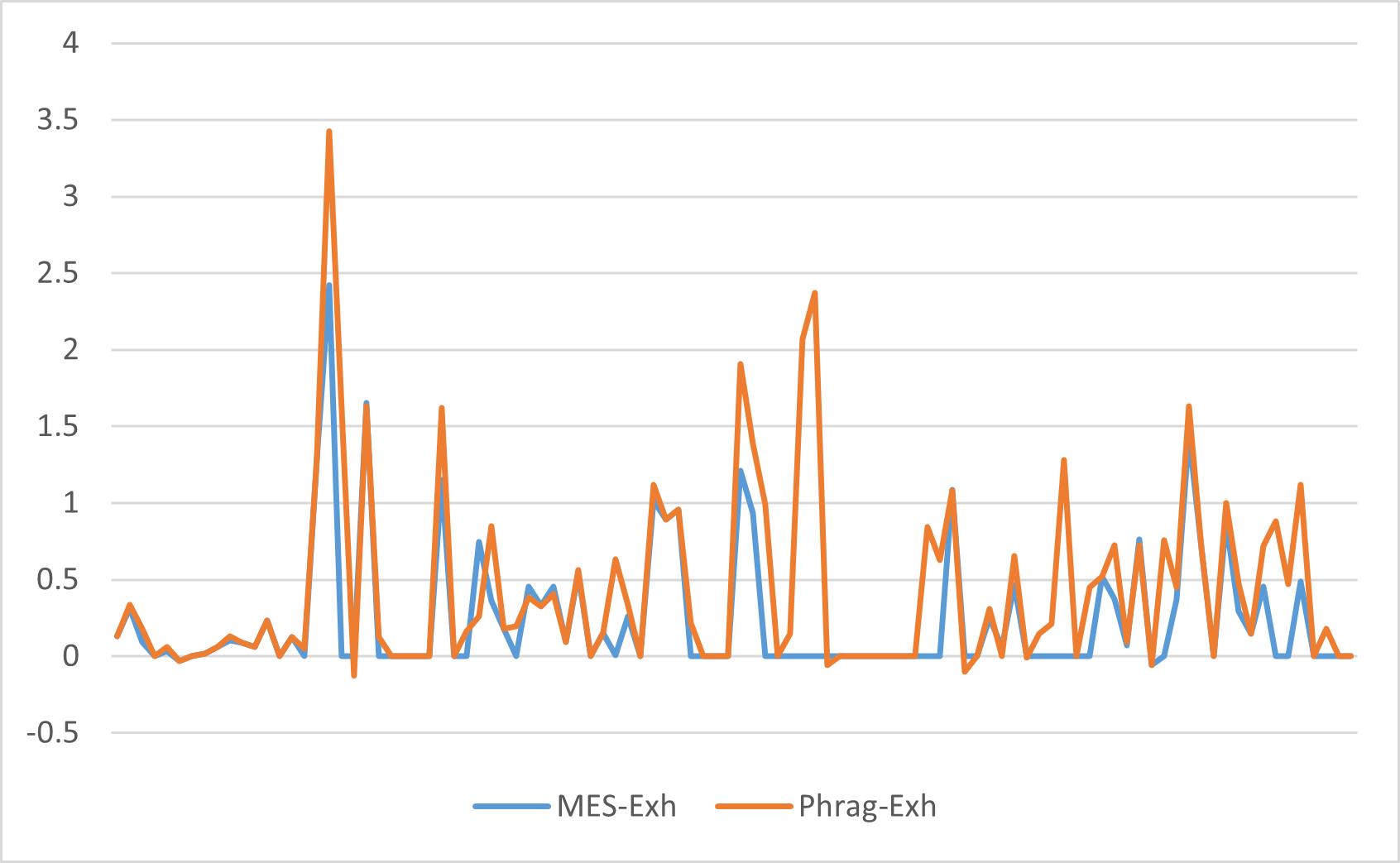}
\end{subfigure}\hfil 
\begin{subfigure}{0.5\textwidth}
  \includegraphics[width=\linewidth]{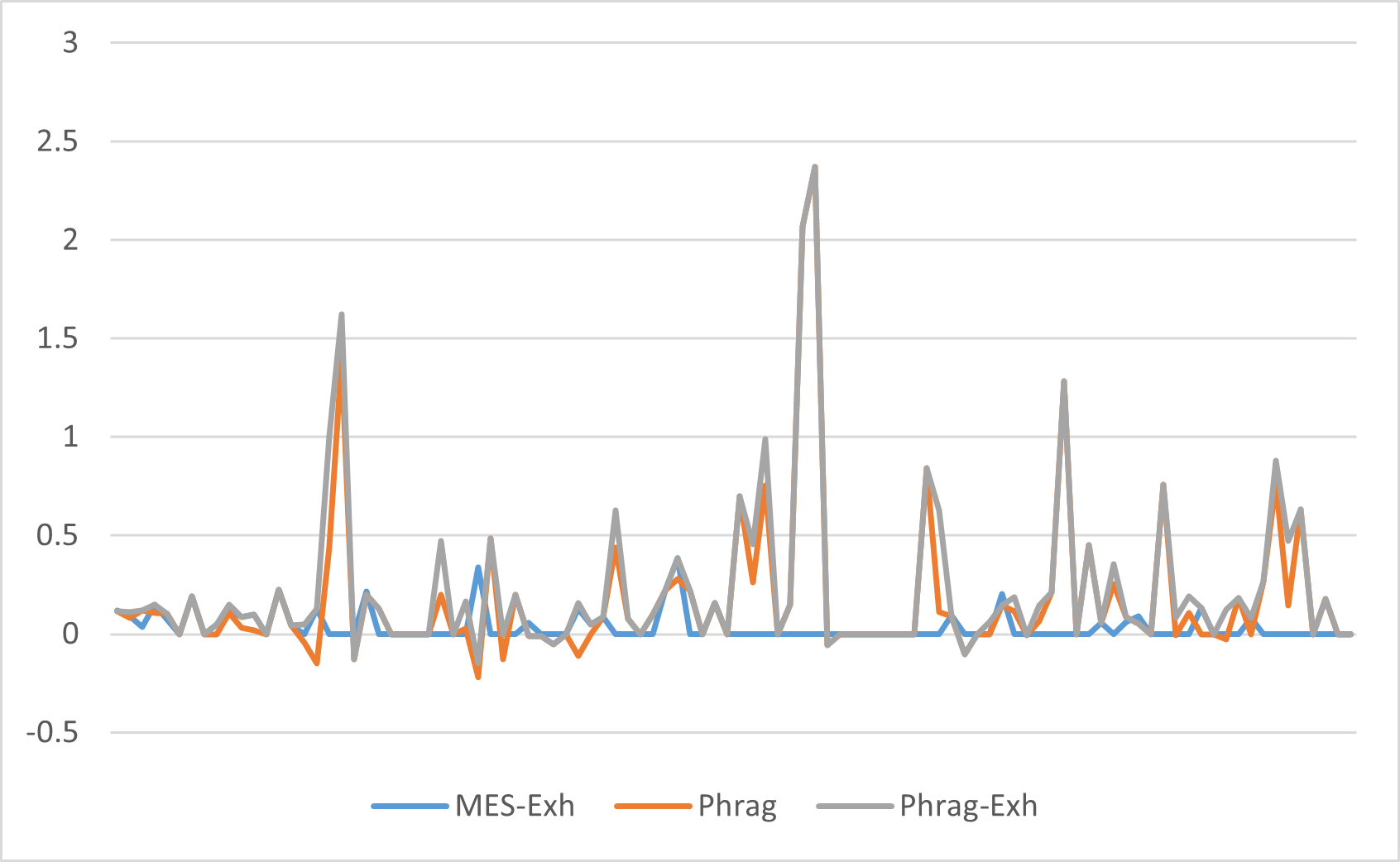}
\end{subfigure}
    \caption{The X-axis corresponds to the dataset number. On the Y-axis, the left graph represents the average proportionality degree of \mes-Exh and Phragm\'{e}n-Exh with greedy approval as the basis (i.e., zero). The right graph represents the average degree of \mes-Exh, Phragm\'{e}n, and Phragm\'{e}n-Exh with \mes as the basis.}
    \label{fig: results_wrt}
\end{figure}

\subsection{Experimental Setup and Methodology}
The first obvious challenge in the design of our experiments is that the number of possible subsets of \proj is exponential in $m$, and hence we cannot possibly calculate proportionality degrees for all of them in reasonable time. To sidestep this, we will use \emph{random sampling}; in particular, for each given subset size $k$, we sample $5000$ sets $T_1 \ldots, T_{5000}$ of $k$ projects uniformly at random from \proj. We consider all possible subset sizes $k \in \{1, \ldots, 15\}$; this is because we observe that for $k > 15$, in almost all datasets, the number of $T$-cohesive groups is $0$. Intuitively, this happens because, as the number of projects in the set increases, the number of voters unanimously agreeing on all of them decreases, and vanishes by the $k=15$ threshold. We then calculate the average proportionality degree over sample subsets $T$ for a given size $k$, and finally take the average of all the possible sizes $k$ to obtain a single representative value for the proportionality degree of each dataset.  

For each of the 100 datasets and each PB rule \pbrule under consideration, we perform the following polynomial time procedure:
\begin{enumerate}
    \item Execute \pbrule on the dataset and compute the satisfaction of each voter from the  outcome of \pbrule.
    \item For each possible subset size $k 
    \in \{1,\ldots,\min(m,15)\}$, we sample exactly $\min\left(\binom{m}{k},5000\right)$ subsets of \proj, each having exactly $k$ distinct projects. The use of the $\min$ operators is because some datasets have fewer than $15$ projects and there could be fewer than $5000$ subsets of projects of a given size $k$.
    \item For each sampled subset $T$, we find the largest set $A(T)$ of voters who approve every project in $T$. Clearly, any $T$-cohesive group of voters must be a subset of $A(T)$.
    \item Note that by \Cref{def: tcohesive}, the cardinality of any $T$-cohesive group of voters has to be at least $n\cdot \cof{T}/\bud$. Thus, we pick exactly $\ceil{n \cdot \cof{T}/{\bud}}$ voters with the least possible satisfaction from $A(T)$ (i.e., the satisfaction of any unpicked voter in $A(T)$ is at least as high as the satisfaction of a picked voter). Clearly, this has to be the $T$-cohesive group of voters with the smallest average satisfaction. Hence, the average satisfaction of this group is the proportionality degree of the sample $T$, $d_{\pbrule}(T)$.
    \item For each subset size $k$, we calculate the average of the proportionality degrees of all the sampled subsets of size $k$, which we denote by $d_k$.
    \item For each dataset, we iterate over all the possible subset sizes up to $\min(m,15)$ and calculate the average proportionality degree of the dataset, i.e., $\frac{\sum_k{d_k}}{\min(m,15)}$.
\end{enumerate}

\subsection{Experimental Results}
First, we illustrate the results for five representative datasets in \Cref{fig: all_datasets} such that the set of rules with the highest average proportionality degree is unique for each of them. We notice that for Dataset 1, \mes-Exh is the (single) best rule, followed by \mes, Phragm\'{e}n-Exh, Phragm\'{e}n, and greedy approval (in that order), whereas for Dataset 2, Phragm\'{e}n-Exh is the (single) best rule, followed by \mes-Exh and Phragm\'{e}n (both very close and \mes-Exh is marginally better), greedy approval, and finally \mes.

\begin{figure}
    \centering
  \scalebox{0.8}{\includegraphics{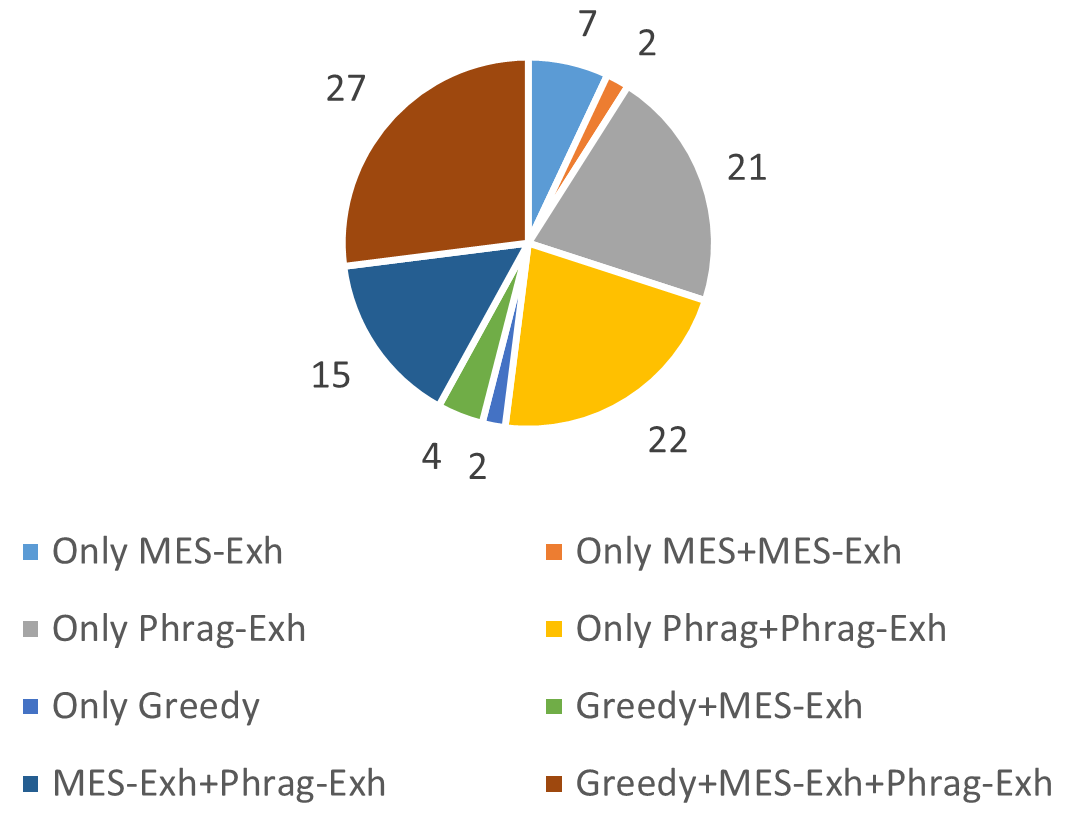}}
  \caption{Pie chart showing the percentage of datasets for which each rule is better than the others.}
  \label{fig: chart}
\end{figure}

Similarly, for Dataset 3, all three rules \mes-Exh, Phragm\'{e}n-Exh, and greedy output the same set of projects. For Dataset 4, both Phragm\'{e}n-Exh and \mes-Exh are best, whereas without exhaustion, Phragm\'{e}n performs much better than \mes. The greedy rule performs worst among all rules. Finally, for Dataset 5, both greedy approval and \mes give the same output which turns out to be better than the output of Phragm\'{e}n-Exh. Overall, \mes-Exh and Phragm\'{e}n-Exh perform well in the experiments, with Phragm\'{e}n performing noticeably well even in the absence of exhaustion. There are also only two out of the 100 datasets for which greedy approval is the single best rule. For both these datasets, \mes and Phragm\'{e}n have the same output and they are only marginally outperformed by the greedy rule. 

Next we present aggregate results over all the datasets. \Cref{fig: results_wrt} shows how the average proportionality degree of different datasets changes with respect to greedy rule and the \mes rule respectively. As shown on the left side, both the rules are significantly better than the greedy rule on most datasets, while Phragm\'{e}n-Exh performs much better than \mes-Exh. Notably, even without exhaustion, the behavior of the different rules is similar. As can be seen on the right side of \Cref{fig: results_wrt}, all of \mes-Exh, Phragm\'{e}n, and Phragm\'{e}n-Exh outperform \mes on most datasets. It is worth highlighting that Phragm\'{e}n performs notably better than \mes-Exh on several datasets, see also the results summarized in \Cref{fig: chart}.

\section{Conclusion and Future Work}
In this paper, we studied the proportionality degree of fundamental rules for approval-based participatory budgeting, focusing on the Method of Equal Shares and Phragm\'{e}n's Sequential Rule, both from a theoretical and from an experimental perspective. The results in our work provide further evidence for the fairness guarantees of these rules, and additional support for their employment in practice over the greedy rule.  

Interestingly, we managed to prove that both rules achieve the same proportionality degree theoretically, albeit with markedly different proof techniques. This leads us to conjecture that in fact both of these rules are optimal for the problem, in the sense that there is no other rule that achieves a better proportionality degree. To this end, below we state a theorem (with a proof in \cref{app:conclusion}) which essentially upper bounds the proportionality degree of any rule by $\frac{\cof{T}}{\max_{t \in T}\cof{t}} - 1$; proving an exactly tight bound (therefore establishing the optimality of the two rules) seems like an interesting technical challenge. 

\begin{theorem}\label{thm:conclusion}
    There exists a PB instance $I$, a set of projects $T \subseteq \proj$ and a $T$-cohesive group $V \subseteq \mathcal{V}(T)$ such that the average satisfaction of the voters in $T$ from the outcome of any PB rule is at most $\frac{\cof{T}}{\max_{t \in T}\cof{t}} - 1$.
\end{theorem}

\bibliographystyle{apalike}
\bibliography{paper.bib}
\clearpage
\appendix
\section*{Appendix}
\section{The proof of \cref{thm:conclusion}}\label{app:conclusion}
\begin{figure}[h!]
    \centering
    \includegraphics[width=0.3\linewidth]{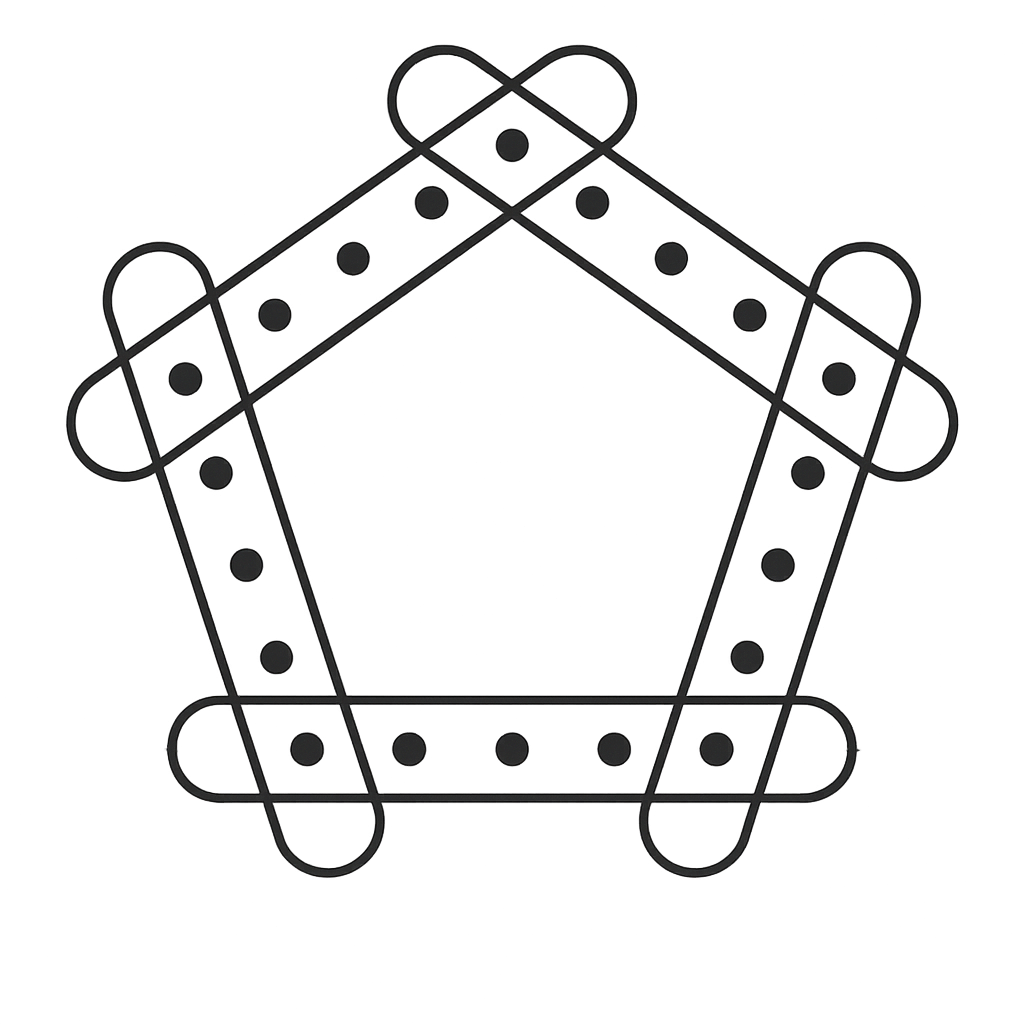}
    \caption{An example instance where $\frac{\bud}{\cof{t_{\min}}} + 1 = 5$. Voters are represented by dots and the boxes represent projects approved by the voters that they contain.}
    \label{fig: aziz-et.al-lower-bound}
\end{figure}
\begin{proof}[Proof of \cref{thm:conclusion}]
    We will consider essentially the same instance as the one presented in \citep{aziz2018complexity} for the case of MWV, with the approval preferences of the voters shown in \Cref{fig: aziz-et.al-lower-bound}. To extend this to the case of PB, we will set the budget $B$ to be divisible by $\min_{t \in T} t$. Let $t_{\min} = \argminaa_{t \in T} \cof{t}$. 
    We will also have $$n = \frac{\bud}{\cof{t_{\min}}}\cdot \lb\frac{\bud}{\cof{t_{\min}}} + 1\rb$$ 
    voters and $\frac{\bud}{\cof{t_{\min}}} + 1$ projects, all of them with cost equal to $\cof{t_{\min}}$. Each voters belongs to at most two groups, depending on the projects that she approves. Each group has size $$|V| = \frac{\bud}{\cof{t_{\min}}} + 1 = \cof{t_{\min}}\cdot \frac{n}{\bud},$$ and thus each of these groups is $(t_{\min})$-cohesive.

    Let $f$ be any PB rule on input instance $I$. Without violating the budget constraints, $f$ can add at most $B/\cof{t_{\min}}$ projects to the proposal. This implies that there exists some $(t_{\min})$-cohesive group for which the average satisfaction will be $$\frac{2}{\lb \frac{\bud}{\cof{t_{\min}}}\rb + 1}$$. 

    Now we extend the construction to the whole set of projects $T$, by replacing project $t_{\min}$ with $|T|$ projects that are approved by the same group of voters (i.e., by any voter $i$ such that $t_{\min} \in A_i)$ and have higher cost than the initial project $t_{\min}$. Next, we set the new budget $\bud^{*}= B \cdot \frac{\cof{T}}{\min_{t \in T}\cof{t}}$, where $\cof{T}$ is divisible by $\min_{t \in T} \cof{t}$. Observe that
    \[
    |V| = \cof{t_{\min}} \cdot \frac{n}{\bud} = \frac{\cof{T}}{\min_{t \in T}\cof{t}} \cdot \cof{t_{\min}} \cdot \frac{n}{\bud^{*}} = \cof{T}\cdot \frac{n}{\bud^{*}}.
    \]
    From the above, it follows that each such group $V$ is $T$-cohesive. Similarly, as in \citep{aziz2018complexity} by selecting projects in a cyclic manner, we can find a set of projects $T$ such that at most $\frac{\cof{T}}{\min_{t \in T}\cof{t}} - 1$ projects are included in the proposal. This is because, it the best case, all of the projects would have cost equal to $\cof{t_{\min}}$, and then there would be at least one group of voters $V'$ such that $\bigcup_{i \in V'}A_i \cap W < \frac{\cof{T}}{\min_{t \in T}\cof{t}}$, i.e., at most $\frac{\cof{T}}{\min_{t \in T}\cof{t}}$ of the projects that they approve are selected in the proposal. Then, the bound still holds for the projects in $T$ with cost higher than $\min_{t \in T}\cof{t}$.

    From the above, it follows that there is a group of voters $V$ such that $\cof{T}\cdot \frac{n}{\bud´} - 2$ voters in $V$ approve at most $\frac{\cof{T}}{\min_{t \in T}\cof{t}} - 1$ projects of the proposal, and the remaining two voters approve at most $\frac{2\cof{T}}{\min_{t \in T}\cof{t}} - 1$ projects of the proposal. Thus, the average satisfaction of the voters in $V$ is at most:

    \begin{small}
    \begin{align*}
        \frac{2\cdot\lb\frac{2\cof{T}}{\min_{t \in T}\cof{t}} - 1\rb + \lb\cof{T}\cdot \frac{n}{\bud´} - 2\rb \cdot \lb\frac{\cof{T}}{\min_{t \in T}\cof{t}} - 1\rb}{\cof{T}\cdot \frac{n}{\bud´}} &= \frac{\frac{2\cof{T}}{\min_{t \in T}\cof{t}}}{\cof{T}\cdot\frac{n}{\bud´}} + \frac{\cof{T}}{\min_{t \in T}\cof{t}} - 1 \\
        &= \frac{2\bud}{n \min_{t \in T}\cof{t}} + \frac{\cof{T}}{\min_{t \in T}\cof{t}} - 1
    \end{align*}
    \end{small}

    Recall that $$n = \frac{\bud}{\min_{t \in T} \cof{t}}\cdot \lb\frac{\bud}{\min_{t \in T} \cof{t}} + 1\rb.$$
    Therefore, for $\bud > \frac{2\min_{t \in T}\cof{t}}{\gamma}$, the average satisfaction of the voters in $V$ is at most than $\frac{\cof{T}}{\min_{t \in T}} - 1 + \gamma$. Then, we may choose $\min_{t \in T}\cof{t}$ such that $\frac{\max_{t \in T}\cof{t}}{\min_{t\in T}\cof{t}} \leq 1 + \delta$ with $\delta \rightarrow 0^{+}$ and we obtain the desired bound.
     
\end{proof}

\end{document}

%% file: figure2.tex
\tikzset{every picture/.style={line width=0.75pt}}

\begin{tikzpicture}[x=0.75pt,y=0.75pt,yscale=-1,xscale=1]


\draw (20,150) node [anchor=west] [align=left, font=\large] {{Round $\pi _{\ell _{1}} :$}};

\draw [rounded corners=5pt] (180, 130) rectangle (290, 170);
\draw (235, 125) node [anchor=south] {$V_{\ell _{1}}$};
\draw (235, 150) node [anchor=center] [font=\small] (idx1) {$0,...,\textcolor[rgb]{0,0.46,1}{i}\textcolor[rgb]{0.56,0.07,1}{\ } ,...,k$};
\draw (235, 175) node [anchor=north] [font=\small] {$any$};

\draw [rounded corners=5pt] (320, 130) rectangle (520, 170);
\draw (420, 125) node [anchor=south] {$V_{\ell _{2}}$};
\draw (420, 150) node [anchor=center] [font=\small] (idx2) {$k+1,...,\textcolor[rgb]{0,0.46,1}{j} \ ,...,n_{V} -1$};
\draw (420, 175) node [anchor=north] [font=\small] {$any$};

\draw (190, 210) node [anchor=north] [font=\large, color={rgb, 255:red, 0; green, 0; blue, 0 }] (eq1) {$p( i) \leq \frac{\max_{t\in T}\text{cost}( t)}{n_{V}} \leq \frac{\max_{t\in T}\text{cost}( t)}{n_{V} -i}$};

\draw (480, 210) node [anchor=north] [font=\large] (eq2) {$p( j) \leq \frac{\max_{t\in T}\text{cost}( t)}{n_{V} -|V_{\ell _{1}} |} \leq \frac{\max_{t\in T}\text{cost}( t)}{n_{V} -j}$};

\draw [->, thick, color={rgb, 255:red, 0; green, 117; blue, 255 }, shorten <= 2pt, shorten >= 5pt] (idx1.south) -- (eq1.north);
\draw [->, thick, color={rgb, 255:red, 0; green, 117; blue, 255 }, shorten <= 2pt, shorten >= 5pt] (idx2.south) -- (eq2.north);

\draw [color={rgb, 255:red, 0; green, 0; blue, 0 }] (20, 255) -- (620, 255);


\draw (20,325) node [anchor=west] [align=left, font=\large] {{Round $\pi _{\ell _{k}} :$}};


\draw [rounded corners=5pt] (120, 305) rectangle (160, 345);
\draw (140, 300) node [anchor=south] {$V_{\ell _{1}}$};
\draw (140, 350) node [anchor=north] [font=\small] {$any$};

\draw [rounded corners=5pt] (170, 305) rectangle (210, 345);
\draw (190, 300) node [anchor=south] {$V_{\ell _{2}}$};
\draw (190, 350) node [anchor=north] [font=\small] {$any$};

\draw (230, 325) node [font=\Huge] {...};

\draw [rounded corners=5pt] (250, 305) rectangle (290, 345);
\draw (270, 300) node [anchor=south] {$V_{\ell _{k-1}}$};
\draw (270, 350) node [anchor=north] [font=\small] {$any$};

\draw [rounded corners=5pt] (310, 305) rectangle (410, 345);
\draw (360, 300) node [anchor=south] {$V_{\ell _{k}}$};
\draw (360, 325) node [anchor=center] [font=\small] (idx3) {$k',...,\textcolor[rgb]{0,0.46,1}{i} \ ,...,\lambda$};
\draw (360, 350) node [anchor=north] [font=\small] {$any$};

\draw [rounded corners=5pt] (430, 305) rectangle (570, 345);
\draw (500, 300) node [anchor=south] {$V \setminus \{V_{\ell _{k}} \cup V_{\overline{T}}\}$};
\draw (500, 325) node [anchor=center] [font=\small] (idx4) {$\lambda+1,...,\textcolor[rgb]{0,0.46,1}{j} \ ,...,n_{V} -1$};
\draw (500, 350) node [anchor=north] [font=\small] {$any$};

\draw (120, 375) .. controls (120, 380) and (205, 380) .. (205, 385) -- (205, 385) .. controls (205, 380) and (290, 380) .. (290, 375);
\draw (205, 390) node [anchor=north] {$V_{\overline{T}}$};

\draw (220, 425) node [anchor=north] [font=\large,color={rgb, 255:red, 0; green, 0; blue, 0 }] (eq3) {$p( i) \leq \frac{\max_{t\in T}\text{cost}( t)}{n_{V} -|V_{\overline{T}} |} \leq \frac{\max_{t\in T}\text{cost}( t)}{n_{V} -i}$};

\draw (520, 425) node [anchor=north] [font=\large] (eq4) {$p( j) \leq \frac{\max_{t\in T}\text{cost}( t)}{n_{V} -|V_{\overline{T}} |-|V_{\ell _{k}} |} \leq \frac{\max_{t\in T}\text{cost}( t)}{n_{V} -j}$};

\draw [->, thick, color={rgb, 255:red, 0; green, 117; blue, 255 }, shorten <= 2pt, shorten >= 8pt] (idx3.south) -- (eq3.north);

\draw [->, thick, color={rgb, 255:red, 0; green, 117; blue, 255 }, shorten <= 2pt, shorten >= 8pt] (idx4.south) -- (eq4.north);

\end{tikzpicture}

%% file: figure1.tex
\tikzset{
    every picture/.style={line width=1.5pt},
    dot/.style={circle, fill=black, inner sep=0pt, minimum size=6pt},
    bluebox/.style={draw=blue!80, line width=1.0pt, rounded corners=6pt},
    label text/.style={font=\large\bfseries} 
}

\begin{tikzpicture}[x=1.4cm, y=0.9cm]

    \draw[bluebox,line width=2pt] (-0.6, -0.6) rectangle (8.5, 0.6);
    \node[label text, blue, anchor=east] at (-0.7, 0) {$T$};

    
    \draw[bluebox] (-0.3, -3.5) rectangle (0.3, 4.5);
    \foreach \y in {-3, -2, -1, 0, 1, 2, 3, 4} { \node[dot] at (0, \y) {}; }
    \node[label text, black] at (0, 4.9) {$Q_0$};
    \node[label text, blue] at (0, -4.0) {$Z_0$};

    \draw[bluebox] (0.7, -3.5) rectangle (1.3, 3.5);
    \foreach \y in {-3, -2, -1, 0, 1, 2, 3} { \node[dot] at (1, \y) {}; }
    \node[label text, black] at (1, 3.9) {$Q_1$};
    \node[label text, blue] at (1, -4.0) {$Z_1$};

    \draw[bluebox] (1.7, -3.5) rectangle (2.3, 2.5);
    \foreach \y in {-3, -2, -1, 0, 1, 2} { \node[dot] at (2, \y) {}; }
    \node[label text, black] at (2, 2.9) {$Q_2$};
    \node[label text, blue] at (2, -4.0) {$Z_2$};

    \node[font=\Huge] at (3.5, 0) {$\dots$};

    
    \draw[bluebox] (4.7, -0.6) rectangle (5.3, 3.5);
    \foreach \y in {0, 1, 2, 3} { \node[dot] at (5, \y) {}; }
    \node[label text, black] at (5, 3.9) {$Q_{n_V-4}$};
    \node[label text, blue] at (5, -1.2) {$Z_{n_V-4}$};

    \draw[bluebox] (5.7, -0.6) rectangle (6.3, 2.5);
    \foreach \y in {0, 1, 2} { \node[dot] at (6, \y) {}; }
    \node[label text, black] at (6, 2.9) {$Q_{n_V-3}$};
    \node[label text, blue] at (6, -1.2) {$Z_{n_V-3}$};

    \draw[bluebox] (6.7, -0.6) rectangle (7.3, 1.5);
    \foreach \y in {0, 1} { \node[dot] at (7, \y) {}; }
    \node[label text, black] at (7, 1.9) {$Q_{n_V-2}$};
    \node[label text, blue] at (7, -1.2) {$Z_{n_V-2}$};

    \node[dot] at (8, 0) {};
    \node[label text, black] at (8.3, 0.3) {$V$};

\end{tikzpicture}

%% file: paper.bib
@article{maly2023core,
  title={The core of an approval-based PB instance can be empty for nearly all cost-based satisfaction functions and for the share},
  author={Maly, Jan},
  journal={arXiv preprint arXiv:2311.06132},
  year={2023}
}

@inproceedings{sreedurga2023individual,
  title={Individual-Fair and Group-Fair Social Choice Rules under Single-Peaked Preferences},
  author={Sreedurga, Gogulapati and Sadhukhan, Soumyarup and Roy, Souvik and Narahari, Yadati},
  booktitle={Proceedings of the 2023 International Conference on Autonomous Agents and Multiagent Systems},
  pages={2872--2874},
  year={2023}
}

@inproceedings{sreedurga2024hybrid,
  title={Hybrid Participatory Budgeting: Divisible, Indivisible, and Beyond},
  author={Sreedurga, Gogulapati},
  booktitle={Proceedings of the 23rd International Conference on Autonomous Agents and Multiagent Systems},
  pages={2480--2482},
  year={2024}
}

@article{bogomolnaia2005collective,
  title={Collective choice under dichotomous preferences},
  author={Bogomolnaia, Anna and Moulin, Herv{\'e} and Stong, Richard},
  journal={Journal of Economic Theory},
  volume={122},
  number={2},
  pages={165--184},
  year={2005},
  publisher={Elsevier}
}

@article{airiau2023portioning,
  title={Portioning using ordinal preferences: Fairness and efficiency},
  author={Airiau, St{\'e}phane and Aziz, Haris and Caragiannis, Ioannis and Kruger, Justin and Lang, J{\'e}r{\^o}me and Peters, Dominik},
  journal={Artificial Intelligence},
  volume={314},
  pages={103809},
  year={2023},
  publisher={Elsevier}
}

@book{wampler2021participatory,
  title={Participatory budgeting in global perspective},
  author={Wampler, Brian and McNulty, Stephanie and McNulty, Stephanie L and Touchton, Michael},
  year={2021},
  publisher={Oxford University Press}
}

@book{lackner2023multi,
  title={Multi-winner voting with approval preferences},
  author={Lackner, Martin and Skowron, Piotr},
  year={2023},
  publisher={Springer Nature}
}

@inproceedings{fain2018fair,
  title={Fair allocation of indivisible public goods},
  author={Fain, Brandon and Munagala, Kamesh and Shah, Nisarg},
  booktitle={Proceedings of the 2018 ACM Conference on Economics and Computation},
  pages={575--592},
  year={2018}
}

@inproceedings{kraiczy2024lower,
  title={A Lower Bound for Local Search Proportional Approval Voting},
  author={Kraiczy, Sonja and Elkind, Edith},
  booktitle={ESA},
  year={2024}
}

@inproceedings{peters2020proportionality,
  title={Proportionality and the limits of welfarism},
  author={Peters, Dominik and Skowron, Piotr},
  booktitle={Proceedings of the 21st ACM Conference on Economics and Computation},
  pages={793--794},
  year={2020}
}

@article{brill2024phragmen,
  title={Phragm{\'e}n’s voting methods and justified representation},
  author={Brill, Markus and Freeman, Rupert and Janson, Svante and Lackner, Martin},
  journal={Mathematical programming},
  volume={203},
  number={1},
  pages={47--76},
  year={2024},
  publisher={Springer}
}

@article{de2022international,
  title={International trends in participatory budgeting},
  author={De Vries, Michiel S and Nemec, Juraj and {\v{S}}pa{\v{c}}ek, David},
  journal={Cham: Palgrave Macmillan},
  year={2022},
  publisher={Springer}
}

@article{thiele1895om,
  title={Om flerfoldsvalg},
  author={Thiele, Thorvald N},
  journal={Oversigt over det Kongelige Danske Videnskabernes Selskabs Forhandlinger},
  volume={1895},
  pages={415--441},
  year={1895}
}

@inproceedings{aziz2018rank,
  title={Rank maximal equal contribution: A probabilistic social choice function},
  author={Aziz, Haris and Luo, Pang and Rizkallah, Christine},
  booktitle={Proceedings of the 32nd AAAI Conference on Artificial Intelligence (AAAI 2018)},
  pages ={910--916},
  year={2018}
}

@inproceedings{faliszewski2023participatory,
  title={Participatory budgeting: data, tools, and analysis},
  author={Faliszewski, Piotr and Flis, Jaros{\l}aw and Peters, Dominik and Pierczy{\'n}ski, Grzegorz and Skowron, Piotr and Stolicki, Dariusz and Szufa, Stanis{\l}aw and Talmon, Nimrod},
  booktitle={Proceedings of the Thirty-Second International Joint Conference on Artificial Intelligence},
  pages={2667--2674},
  year={2023}
}

@inproceedings{aziz2018proportionally,
author = {Aziz, Haris and Lee, Barton E. and Talmon, Nimrod},
title = {Proportionally Representative Participatory Budgeting: Axioms and Algorithms},
booktitle = {Proceedings of the 17th International Conference on Autonomous Agents and MultiAgent Systems},
pages = {23–31},
year = {2018}
}

@inproceedings{fairstein2022welfare,
author = {Fairstein, Roy and Vilenchik, Dan and Meir, Reshef and Gal, Kobi},
title = {Welfare vs. Representation in Participatory Budgeting},
booktitle = {Proceedings of the 21st International Conference on Autonomous Agents and Multiagent Systems},
pages = {409–417},
year = {2022}
}

@article{skowron2020participatory,
  title={Participatory budgeting with cumulative votes},
  author={Skowron, Piotr and Slinko, Arkadii and Szufa, Stanis{\l}aw and Talmon, Nimrod},
  journal={arXiv preprint arXiv:2009.02690},
  year={2020}
}

@inproceedings{sanchez2017proportional,
  title={Proportional justified representation},
  author={S{\'a}nchez-Fern{\'a}ndez, Luis and Elkind, Edith and Lackner, Martin and Fern{\'a}ndez, Norberto and Fisteus, Jes{\'u}s A and Val, Pablo Basanta and Skowron, Piotr},
  booktitle={AAAI 2017},
  year={2017}
}

@article{aziz2017justified,
  title={Justified representation in approval-based committee voting},
  author={Aziz, Haris and Brill, Markus and Conitzer, Vincent and Elkind, Edith and Freeman, Rupert and Walsh, Toby},
  journal={Social Choice and Welfare},
  volume={48},
  number={2},
  pages={461--485},
  year={2017},
  publisher={Springer}
}

@inproceedings{aziz2018complexity,
  title={On the Complexity of Extended and Proportional Justified Representation},
  author={Aziz, Haris and Elkind, Edith and Huang, Shenwei and Lackner, Martin and Fern{\'a}ndez, Luis S{\'a}nchez and Skowron, Piotr},
  booktitle={AAAI 2018},
  pages={902--909},
  year={2018}
}

@inproceedings{fluschnik2019fair,
  title={Fair knapsack},
  author={Fluschnik, Till and Skowron, Piotr and Triphaus, Mervin and Wilker, Kai},
  booktitle={Proceedings of the 33rd AAAI Conference on Artificial Intelligence (AAAI 2019)},
  pages={1941--1948},
  year={2019}
}

@article{phragmen1894methode,
  title={Sur une m{\'e}thode nouvelle pour r{\'e}aliser, dans les {\'e}lections, la repr{\'e}sentation proportionnelle des partis},
  author={Phragm{\'e}n, Edvard},
  journal={{\"O}fversigt af Kongliga Vetenskaps-Akademiens F{\"o}rhandlingar},
  volume={51},
  number={3},
  pages={133--137},
  year={1894}
}

@article{laruelle2021voting,
  title={Voting to select projects in participatory budgeting},
  author={Laruelle, Annick},
  journal={European Journal of Operational Research},
  volume={288},
  number={2},
  pages={598--604},
  year={2021}
}

@misc{benade2018efficiency,
  title={Efficiency and usability of participatory budgeting methods},
  author={Benade, Gerdus and Itzhak, Nevo and Shah, Nisarg and Procaccia, Ariel D},
  howpublished = {\url{https://www.cs.toronto.edu/~nisarg/papers/pb_usability.pdf}},
  year={2018}
}

@inproceedings{jain2020participatoryg,
  title={Participatory Budgeting with Project Groups},
  author={Jain, Pallavi and Sornat, Krzysztof and Talmon, Nimrod and Zehavi, Meirav},
  booktitle={Proceedings of the 30th International Joint Conference on Artificial Intelligence (IJCAI 2021)},
  pages={276-282},
  year={2021}
}

@article{freeman2021truthful,
  title={Truthful aggregation of budget proposals},
  author={Freeman, Rupert and Pennock, David M and Peters, Dominik and Vaughan, Jennifer Wortman},
  journal={Journal of Economic Theory},
  volume={193},
  pages={105234},
  year={2021},
  publisher={Elsevier}
}

@article{peters2021proportional,
  title={Proportional participatory budgeting with additive utilities},
  author={Peters, Dominik and Pierczy{\'n}ski, Grzegorz and Skowron, Piotr},
  journal={Advances in Neural Information Processing Systems},
  volume={34},
  pages={12726--12737},
  year={2021}
}

@inproceedings{talmon2019framework,
  title={A framework for approval-based budgeting methods},
  author={Talmon, Nimrod and Faliszewski, Piotr},
  booktitle={Proceedings of the 33rd AAAI Conference on Artificial Intelligence (AAAI 2019)},
  pages={2181--2188},
  year={2019}
}

@article{benade2021preference,
  title={Preference elicitation for participatory budgeting},
  author={Benade, Gerdus and Nath, Swaprava and Procaccia, Ariel D and Shah, Nisarg},
  journal={Management Science},
  volume={67},
  number={5},
  pages={2813--2827},
  year={2021},
  publisher={INFORMS}
}

@article{sreedurga2022indivisible,
  title={Indivisible participatory budgeting under weak rankings},
  author={Sreedurga, Gogulapati and Narahari, Yadati},
  journal={arXiv preprint arXiv:2207.07981},
  year={2022}
}

@inproceedings{aziz2021proportionally,
  title={Proportionally Representative Participatory Budgeting with Ordinal Preferences},
  author={Aziz, Haris and Lee, Barton E},
  booktitle={Proceedings of the 35th AAAI Conference on Artificial Intelligence (AAAI 2021)},
  pages={5110--5118},
  year={2021}
}

@article{duddy2015fair,
  title={Fair sharing under dichotomous preferences},
  author={Duddy, Conal},
  journal={Mathematical Social Sciences},
  volume={73},
  pages={1--5},
  year={2015},
  publisher={Elsevier}
}

@String{Macmillan = "Macmillan" }

@String{Springer = "Springer-Verlag" }

@inproceedings{aziz2019fair,
  title={Fair mixing: the case of dichotomous preferences},
  author={Aziz, Haris and Bogomolnaia, Anna and Moulin, Herv{\'e}},
  booktitle={Proceedings of the 2019 ACM Conference on Economics and Computation (EC 2019)},
  pages={753--781},
  year={2019}
}

@inproceedings{fain2016core,
  title={The core of the participatory budgeting problem},
  author={Fain, Brandon and Goel, Ashish and Munagala, Kamesh},
  booktitle={Proceedings of the 12th International Conference on Web and Internet Economics (WINE 2016)},
  pages={384--399},
  year={2016}
}

@article{goel2019knapsack,
  title={Knapsack voting for participatory budgeting},
  author={Goel, Ashish and Krishnaswamy, Anilesh K and Sakshuwong, Sukolsak and Aitamurto, Tanja},
  journal={ACM Transactions on Economics and Computation (TEAC)},
  volume={7},
  number={2},
  pages={1--27},
  year={2019},
  publisher={ACM New York, NY, USA}
}

@incollection{aziz2021participatory,
  title={Participatory budgeting: Models and approaches},
  author={Aziz, Haris and Shah, Nisarg},
  booktitle={Pathways Between Social Science and Computational Social Science},
  pages={215--236},
  year={2021},
  publisher={Springer}
}

@inproceedings{rey2020designing,
  title={Designing participatory budgeting mechanisms grounded in judgment aggregation},
  author={Rey, Simon and Endriss, Ulle and de Haan, Ronald},
  booktitle={Proceedings of the 17th International Conference on Principles of Knowledge Representation and Reasoning (KR 2020)},
  pages={692--702},
  year={2020}
}

@article{fairstein2021proportional,
  title={Proportional Participatory Budgeting with Substitute Projects},
  author={Fairstein, Roy and Meir, Reshef and Gal, Kobi},
  journal={arXiv preprint arXiv:2106.05360},
  year={2021}
}

@inproceedings{sreedurga2022maxmin,
  title={Maxmin Participatory Budgeting},
  author={Sreedurga, Gogulapati and Bhardwaj, Mayank Ratan and Narahari, Y},
  booktitle={Proceedings of the 31st International Joint Conference on Artificial Intelligence (IJCAI 2022)},
  pages={489--495},
  year={2022}
}

@inproceedings{baumeister2020irresolute,
  title={Irresolute approval-based budgeting},
  author={Baumeister, Dorothea and Boes, Linus and Seeger, Tessa},
  booktitle={Proceedings of the 19th International Conference on Autonomous Agents and MultiAgent Systems (AAMAS 2020)},
  pages={1774--1776},
  year={2020}
}

@inproceedings{brill2023proportionality,
  title={Proportionality in Approval-Based Participatory Budgeting},
  author={Brill, Markus and Forster, Stefan and Lackner, Martin and Maly, Jan and Peters, Jannik},
  booktitle={Proceedings of the 37th AAAI Conference on Artificial Intelligence (AAAI 2023)},
  pages={5524--5531},
  year={2023}
}

@inproceedings{los2022proportional,
  title={Proportional budget allocations: Towards a systematization},
  author={Los, Maaike and Christoff, Zo{\'e} and Grossi, Davide},
  booktitle={Proceedings of the 31st International Joint Conference on Artificial Intelligence (IJCAI 2022)},
  pages={398--404},
  year={2022}
}

@inproceedings{skowron2021proportionality,
  title={Proportionality degree of multiwinner rules},
  author={Skowron, Piotr},
  booktitle={Proceedings of the 22nd ACM Conference on Economics and Computation},
  pages={820--840},
  year={2021}
}

@incollection{kilgour2010approval,
  title={Approval balloting for multi-winner elections},
  author={Kilgour, D Marc},
  booktitle={Handbook on approval voting},
  pages={105--124},
  year={2010},
  publisher={Springer}
}
